\pdfoutput=1

\documentclass[12pt]{article}
\usepackage[paper=letterpaper,dvips,top=2.5cm,left=2.5cm,right=2.5cm,foot=1.5cm,bottom=2.5cm]{geometry}
\usepackage{graphicx}
\usepackage{amssymb}
\usepackage{amsmath}
\usepackage{amsthm}
\usepackage{mathtools}
\usepackage{bbold}
\usepackage{amsmath, amsthm, amssymb}
\usepackage{bbm,accents }
\usepackage[authoryear]{natbib}
\usepackage{lipsum}

\usepackage{algorithm}
\usepackage{algpseudocode}
\usepackage{appendix}

\usepackage{graphicx}
\usepackage{subcaption}
\usepackage{mwe}
\newtheorem{lemma}{Lemma}
\newtheorem{theorem}{Theorem}

\newtheorem{prop}{Proposition}

\newcommand\norm[1]{\left\lVert#1\right\rVert}

\newcommand{\Tau}{\mathrm{T}}
\newcommand{\supp}{\text{supp}}

\def\multiset#1#2{\ensuremath{\left(\kern-.3em\left(\genfrac{}{}{0pt}{}{#1}{#2}\right)\kern-.3em\right)}}

\usepackage{color}
\usepackage[bookmarks=false]{hyperref}

\definecolor{darkblue}{rgb}{0.0,0.0,0.6}
\hypersetup{
	colorlinks=true, %set true if you want colored links
	linktoc=all,     %set to all if you want both sections and subsections linked
	linkcolor=darkblue,  %choose some color if you want links to stand out
	citecolor=darkblue,
	urlcolor=darkblue,
}

\begin{document}

\title{Cooperation in Small Groups - an Optimal Transport Approach\footnote{I am indebted to  Truman Bewley, John Geanakopolos, Xiangliang Li and Larry Samuelson for inspiring discussions and their encouragements. I am grateful to Ian Ball, Laura Doval, Michael Greinecker, Ryota Iijima, Masaki Miyashita, Allen Wong, Weijie Zhong and participants at conferences and	seminars for helpful comments. I am responsible for all remaining mistakes and typos. First version: March 2017. }}
\author{Xinyang Wang\footnote{Department of Economics, Yale University, email: \href{mailto:xinyang.wang@yale.edu}{xinyang.wang@yale.edu} }}
\maketitle

\abstract
If agents cooperate only within small groups of some bounded sizes, is there a way to partition the population into small groups such that no collection of agents can do better by forming a new group? This paper revisited f-core in a transferable utility setting. By providing a new formulation to the problem, we built up a link between f-core and the transportation theory. Such a link helps us to establish an exact existence result, and a characterization result of f-core for a general class of agents, as well as some improvements in computing the f-core in the finite type case.

\section{Introduction}

In this paper, we study a continuum of players form small groups in order to share group surpluses. Group sizes are bounded by natural numbers or percentiles, and group surpluses are determined by the types of its members. We wish to know whether there is a stable state in this game: that is, whether we can partition the continuum of players into small groups such that agents have a way to share group surpluses and no coalition of players can do better by forming a new group. We use the term stable assignments to denote such partitions. Therefore, our question is whether a stable assignment exists. It is worth noting that when the group size is exactly two, the problem is a (roommate) matching problem.

In the literature, there are some partial answers to the question we posed: when group sizes are bounded by natural numbers, \cite{KanekoWooders96} proved an approximately stable assignment exists\footnote{In their definition of stable assignment, there is an additional approximate feasibility condition. In addition, their framework is for games with non-transferable utility. In this paper, when we mention their work, we talk about the application of their result to a game with transferable utility.}. It is only known that the approximation notion can be dropped when the continuum of players share a finite number of types\footnote{See  \cite{wooders2012theory} for a proof.}.  When group sizes are bounded by percentiles, \cite{Schmeidler72} proved, in exchange economies, core allocations are not blocked by any group of epsilon sizes. However, this observation alone will not lead to the existence of a stable assignment, since core allocations  cannot  be achieved by  reallocation in small groups. 

As a result, to my knowledge, when group sizes are bounded by either natural numbers or percentiles, there is no existence result for general models. To make matters worse, even though we know the existence of a stable assignment when group sizes are finitely bounded and the type space is finite, it is not computationally feasible to use the current method, linear programming, to find a stable assignment, as the number of group types becomes astronomical even when group sizes are very small. For example, when there are 1000 types of players and every group contains up to 4 players, there are $\sum_{n=1}^4 \multiset{1000}{n}\sim 4.2\times 10^{10}$ types of groups in this game. If we use linear programming to find a stable assignment, we need to solve a maximization problem with $4.2\times 10^{10}$ unknowns, which is not a feasible task computationally.

In this paper, conceptually, we prove the existence of a stable assignment for general type spaces and general surplus functions, when group sizes are bounded by either natural numbers or percentiles. Furthermore, when group sizes are bounded by natural numbers, our work provides a parallel yet simpler formulation for the classical assignment problem; when group sizes are bounded by percentiles, our model is new.

Computationally, when there are finitely many types of agents in the market, our method provides the computational feasibility for finding stable assignments by reducing the number of unknowns from about $|I|^N$, in linear programming, to about $|I|$, where $|I|$ is the number of types, $N$ is maximum group size and $|I|$ is much larger than $N$. In particular, for games with 1000 types of players and group sizes are up to 4, our method reduces the number of unknowns in the maximization problem from about $4.2\times 10^{10}$ to $25000$.\footnote{See \hyperref[subsubsection: production]{Section \ref{subsubsection: production}} for more details.} 

Methodologically, we prove the existence results by connecting the stability problem to a multi- or ``continuum"- marginal transport problem. For games with finite-size groups, our key observation is that any stable assignment can be identified by a {\it symmetric} transport plan in some replicated agent spaces. When the maximum group size is $N$, we replicate the agents space $N!$ many times. We choose the number $N!$ as any group, containing no more than $N$ players, can be represented by a vector of length $N!$ via replications. As a consequence of our observation, we are able to reformulate the welfare maximization problem as a symmetric transport problem. This reformulation helps us in two ways. First, the duality theorem of the multi-marginal transport problems helps us prove the existence of a stable assignment. Second, as \cite{friesecke2018breaking} pointed out, the symmetric structure created by our identification trick helps us reduce the number of unknowns. For games with positive-size groups, we prove the existence of a stable assignment by extending Kantorovich-Koopmans' duality theorem to a ``continuum"-marginal setting.

This paper is organized as follows. In  \hyperlink{section.2}{Section 2}, we study games with finite-size groups. In \hyperlink{section.3}{Section 3}, we study games with positive-size groups. We review related literature at the end of each section. In \hyperlink{section.4}{Section 4}, we summarize our results.

\section{Games with Finite-Size Groups} \label{section: finitegames}
In this section, we study transferable utility games with a continuum of agents who form small groups in order to share group surpluses. Group sizes are bounded by natural numbers. We use a natural number $N\ge 2$ to denote the upper bound on group sizes and a natural number $N'\le N$ to denote the lower bound on group sizes.

This section is organized as follows.  In \hyperlink{subsection.2.1}{Section 2.1}, we describe the model. In \hyperlink{subsection.2.2}{Section 2.2}, we state our results. In \hyperlink{subsection.2.3}{Section 2.3}, we prove our results. In \hyperlink{subsection.2.4}{Section 2.4}, we study four examples. In \hyperlink{subsection.2.5}{Section 2.5}, we discuss three applications of our results. Lastly, in \hyperlink{subsection.2.6}{Section 2.6}, we review the literature.

\subsection{Model}
We study a cooperative (transferable utility) game denoted by the tuple $((I,\mu),s, N', N)$. Here, the natural numbers $N'$ and $N$ correspond to the lower and the upper bound on group sizes.

\subsubsection{Type Space}
The {\it type space} of players is summarized by a measure space $(I,\mu)$. In particular, the set $I$ is a compact
metric space\footnote{\label{footnote: cptassumption}The compactness of the type space can be relaxed with no essential changes in the rest of the paper. See \hyperref[subsection: remarks]{Section \ref{subsection: remarks}} for more details.} representing a set of players' types. The measure $\mu\in \mathcal{M}_+(I)$ is a non-negative  finite Borel measure on $I$ representing the distribution of players' types. 

There are two canonical examples of the type space. Firstly, $I$ is a finite set and $\mu$ is an $|I|$-dimensional real vector with positive entries. In this finite type case, the $i$-th coordinate of $\mu$, $\mu(i)>0$, denotes the mass of type $i$ players in this game. We will use the finite type case to explain some definitions in the following subsections. Secondly, $I$ is the unit interval $[0,1]$ and $\mu$ is some probability measure on $I$ with or without atoms. %Thirdly, $I$ is the real line $\mathbb{R}$ and $\mu$ is the normal distribution on $\mathbb{R}$. This non-compact type space can help us to analyze models with normal distributed noises\footnote{For example, an exchange economy with agents having normal distributed endowments.}.

% Probabilty \mu is not essential

\subsubsection{Groups} \label{subsection: groups1}
Groups are the units in which players interact with each other. Since we only distinguish players according to their types, we can only distinguish groups according to group types. To be simple, we abuse the language by calling group types groups.

Formally, for any permitted group size $n\in\mathbb{N}$ such that $N'\le n\le N$, {\it the set of $n$-person groups} $\mathcal{G}_n$ consists of all multisets of cardinality $n$ with elements taken from $I$. That is,
$$\mathcal{G}_n=\left\{G:I\rightarrow \mathbb{N} \left\vert \sum_{i\in I}G(i)=n\right. \right\}\footnote{\label{footnote:infinitesum} Implicitly, for any function $G\in\mathcal{G}_n$, $G$ has nonzero values at most finitely many points in $I$. That is, the support of G, $\supp(G)=\{i\in I: G(i)\neq 0\}$, is a finite set. The summation $\sum_{i\in I} G(i)$ is thus defined as $\sum_{i\in \supp(G)} G(i)$.}$$
In any $n$-person group $G\in \mathcal{G}_n$, there are $G(i)$ type $i$ players, for any type $i\in I$. 

{\it The set of groups} $\mathcal{G}$ consists of all groups of permitted sizes. Therefore,
$$\mathcal{G}=\bigcup_{n=N'}^N \mathcal{G}_n$$
Naturally, a {\it group} is an element in $\mathcal{G}$. All groups in this section consist of finitely many types of players. Groups consisting of infinitely many types of players will be discussed in \hyperlink{section.3}{Section 3}.

Next, we identify $n$-person groups by the equivalence classes on the product space $I^n$. For any natural number $n\ge 2$, we define an equivalence relation $\sim_n$ on $I^n$: for any type lists $(i_1,...,i_n), (j_1,...,j_n)\in I^n$,
$$(i_1,...,i_n)\sim_n (j_1,...,j_n)\Longleftrightarrow j_k=i_{\sigma(k)}~\text{for~all~}1\le k\le n\text{,~for~some~permutation~}\sigma\in S_n$$
That is, two type lists $(i_1,...,i_n), (j_1,...,j_n)$ in $I^n$ are equivalent if they are the same up to some index permutation. It is easy to verify there is a bijection  between the set of $n$-person groups $\mathcal{G}_n$ and the set of equivalence classes $I^n/\sim_n$\footnote{The bijective map is defined by $T: I^n/\sim_n\rightarrow \mathcal{G}_n$ such that $T([i_1,...,i_n])=G$, where $G: I\rightarrow \mathbb{N}$ is defined by $G(i)=|\{1\le k\le n: i_k=i\}|$ for all $i\in I$. }.  i.e. $\mathcal{G}_n \simeq I^n/\sim_n$. Therefore, we write an $n$-person group $G\in \mathcal{G}_n$ as
$$G=[i_1,...,i_n]$$
by listing all its members' types with repetitions. By identifying $n$-person groups by equivalence classes $I^n/\sim_n$, we know $\mathcal{G}_n$ is metrizable under the quotient topology.\footnote{See  \hyperref[appendix:equivalenceclass]{Appendix \ref{appendix:equivalenceclass}} for a proof.}.

In the finite type case, our model can be described in the language of hypergraph theory: the type set $I$ corresponds to vertices of a hypergraph, and the set of $n$-person groups $\mathcal{G}_n$ corresponds to $n$-uniform hyperedges in the hypergraph. It is well known that the number of $n$-person groups, or $n$-uniform hyperedges, is given by,
\begin{equation}
|\mathcal{G}_n|=\multiset{|I|}{n} =  {|I|+n-1 \choose n}=\frac{(|I|+n-1)!}{n!(|I|-1)!} 
\end{equation}
In particular, for any fixed group size $n$, there are $\Theta(I^n)$ types of $n$-person groups.

\subsubsection{Surplus Function}
A surplus function specifies the total amount of surplus group members can share. Formally, a {\it surplus function} $s$ is a non-negative valued function on the set of groups $\mathcal{G}$. Its restriction on $\mathcal{G}_n$, $s_n=s|_{\mathcal{G}_n}$, is a surplus function on $n$-person groups. We impose the following assumption on the surplus function: 
\begin{enumerate}
	\item[$(A1)$]\label{A1} for any permitted group size $N'\le n\le N$, $s_n$ is continuous in $\mathcal{G}_n$\footnote{\label{footnote: ctsassumption}The continuity assumption can be replaced by an upper semi-continuity assumption with no essential changes on the rest of this paper. See \hyperref[subsection: remarks]{Section \ref{subsection: remarks}} for more details.}. i.e. for any $[i_1^k,...,i_n^k]\rightarrow [i_1,...,i_n]$ in $\mathcal{G}_n$, 
	$$\lim_{k\rightarrow \infty} s_n([i_1^k,...,i_n^k])= s_n([i_1,...,i_n)])$$
	%\item[$(A1)'$] $s_n$ is continuous in $\mathcal{G}_n$. i.e. for any $(i_1^k,...,i_n^k)\rightarrow (i_1,...,i_n)$ in $\mathcal{G}_n$, 
	%$$\lim_{k\rightarrow \infty} s_n([i_1^k,...,i_n^k])= s_n([i_1,...,i_n)])$$
	%\item[$(A2)$]\label{A2} for any permitted group size $N'\le n\le N$, there is a continuous function $a_{n}\in L^1(I,\mu)$ such that
	%$$s_n([i_1,...,i_n])\le \sum_{k=1}^n a_{n}(i_k)$$
\end{enumerate}
\hyperref[A1]{Assumption (A1)} assumes, for any permitted group size $n$, the surplus function on $n$-person groups is continuous. %In contrast, Assumption $(A1)'$ assumes, for any possible group sizes n, the surplus function on $n$-person groups is continuous. It will give us the full regularity of players' payoffs in the stable situation.
%The second assumption states the consistency requirement: as $(i_1,..., i_n)$ and $(i_{\sigma(1)},...,i_{\sigma(n)})$ denote the same set of agents, they have the same surplus.
%\hyperref[A2]{Assumption (A2)} assumes, for  any permitted group sizes $n$, the surplus function on $n$-person groups is bounded from above by some integrable continuous function. This assumption is trivially satisfied when $s_n$ is a bounded function or when the type set is compact.  Under \hyperref[A2]{Assumption (A2)}, there is a well-defined integrable imputation that is not blocked by any coalition. 
In the finite type case, this assumption impose no restriction on the surplus function.

In addition, we remark that we assume no relation between surplus functions on different group sizes. In particular, the surplus function need not be super-additive: the departure of an agent from a group might increase the surplus of the remaining group members. The non super-additivity helps us to analyze examples such as exchange economies with consumption externalities\footnote{An example is given in \hyperlink{subsubsection.2.4.4}{Section 2.4.4}.}. %Secondly, due to the fact that an indicator function of any closed set is upper semi-continuous, the discontinuities of the surplus function appears naturally in many economic applications of the model. For instance, we can use this model to study the matching problem, in which the surpluses of groups consisting of two people of same sex are zero and the surpluses of groups consisting of two people of different sex are no less than 1. That is, by considering semi-continuous surplus functions, our model can help to study models with more restricted group structures.

%For instance, there are a half unit mass of homogeneous sellers, who have 1 dollar initially, and a half unit mass of heterogeneous buyers, who have 0 dollar initially. In this market, the type space can be represented by the continuum $[0,1]$ with a probability measure $1/2\delta_0+1/2 \mu_{Leb}$ on it, where $\mu_{Leb}$ is the Lebesgue measure. If the surplus function of a group is defined by the sum of the initial dollars agents have, the surplus function is discontinuous at the origin.

%Indeed, the surplus of a group could have a jump (upward or downward) when one agent outside the group joins or one agent in the group leaves. %\footnote{On a separate note, one can define function v to be a real-valued function on the equivalent class of $I^n$, but we do not use that formalization here.} 

\subsubsection{Assignments} \label{subsection: assignments1}

An assignment is a partition of the continuum of players into groups of permitted sizes. 

Rather than studying the partition directly, we define an assignment by a statistical representation of the partition, in which the mass of each group in the partition is specified. By using this statistical representation, we simplify the classical definition of assignment in the literature\footnote{In \cite{kaneko1986core} and \cite{KanekoWooders96}, a partition $p$ of the player space $I$ is defined to be measure-consistent if there is a partition of $I$ into $N$ measurable sets $I_1,...,I_N$ and each set $I_{n}$ has a partition, consisting of measurable subsets $\{I_{n1},...,I_{nn}\}$, with the following property: there are measure preserving isomorphisms $\psi_{n1},...,\psi_{nn}$ from $I_{n1}$ to $I_{n1},...,I_{nn}$, respectively, such that $\psi_{n1}$ is the identity map and $\{\psi_{n1},...,\psi_{nn}\}\in p$ for all $i\in I_{k1}$.} and also obtain some measure structure in the definition.

To describe this statistical representation, we first need to explore the structure of $\mathcal{G}_n$, the set of $n$-person groups. In particular, we will define a collection of partitions of $\mathcal{G}_n$:  for every measurable subset $A\subset I$ and every natural number
 $0\le k\le n$, the set $\mathcal{G}_n(A,k)$ is defined to be a set consisting of all $n$-person groups in which there are $k$ agents whose types are in the set A. That is,
$$\mathcal{G}_n(A,k)=\left\{G\in\mathcal{G}_n: \sum_{i\in A} G(i)=k\right\}\footnote{Similar to  \hyperref[footnote:infinitesum]{Footnote \ref{footnote:infinitesum}}, $\sum_{i\in A} G(i)$ is defined  to be $\sum_{i\in A\cap \supp(G)} G(i)$.}$$
It is routine to check that, for every measurable $A\subset I$, $\{\mathcal{G}_n(A,k):0\le k\le n\}$ is a partition of $\mathcal{G}_n$. That is, the collection of partitions we defined is $\{\{\mathcal{G}_n(A,k):0\le k\le n\}: \text{~measurable~} A\subset I  \}$. Moreover, $\mathcal{G}_n(A,k)$ is a measurable set in $\mathcal{G}_n=I^n/\sim_n$.\footnote{See  \hyperref[appendix: Gnkmeasurable]{Appendix \ref{appendix: Gnkmeasurable}} for a proof.}.

Now, we define assignments. An {\it assignment} is a tuple $\tau=(\tau_{N'},...,\tau_N)$, where $\tau_n$ is a non-negative measure on $\mathcal{G}_n$ satisfying the following consistency condition: for any measurable $A\subset I$,
\begin{equation}\label{eqn: consistency}
\sum_{n=N'}^N \sum_{k=0}^n k\tau_n (\mathcal{G}_n(A,k))=\mu(A)
\end{equation}
In this equation, $\tau_n (\mathcal{G}_n(A,k))$ is the total mass of all $n$-person groups containing exactly $k$ players whose types are in $A$. Thus, $k\tau_n (\mathcal{G}_n(A,k))$ is the total mass of players, whose types are in $A$, that are assigned to some $n$-person groups containing exactly $k$ players whose types are in set $A$. Summing over $k$, $\sum_{k=0}^n k\tau_n (\mathcal{G}_n(A,k))$ is the total mass of players, whose types are in $A$, that are assigned to some $n$-person groups. Therefore, $\sum_{n=N'}^N \sum_{k=0}^n k\tau_n (\mathcal{G}_n(A,k))$ is the total mass of players, whose types are in $A$, that are assigned by the assignment. Since all players need to be assigned by the assignment, this total mass is equal to $\mu(A)$, which is the total mass of players whose types are in set $A$. Consequently, the consistency condition means all players are assigned to some group in the partition specified by the assignment.

In the finite type case, an assignment is a vector $\tau$ with $|\mathcal{G}|$ non-negative real entries. For any group $G\in \mathcal{G}$, $\tau(G)$ is the mass of group $G$ in the partition represented by $\tau$. In this case, the consistency condition is: for any type $i\in I$,
$$\sum_{k=0}^N \sum_{G\in \mathcal{G}(i,k)}k\tau (G)=\mu(i)$$
where $\mathcal{G}(i,k)$ is the set of groups containing exactly $k$ type $i$ players.

We finish this subsubsection by defining the following notions.

Firstly, we use a set $\Tau$ to denote {\it the set of assignments}. That is,
$$\Tau=\left\{\tau=(\tau_{n})_{N'\le n\le N}: \tau_{n}\in \mathcal{M}_+(\mathcal{G}_n), \forall n, \sum_{n=N'}^N \sum_{k=0}^n k\tau_n (\mathcal{G}_n(A,k))=\mu(A), \forall A\subset I  \right\}$$

Secondly, given any group  $G=[i_1,...,i_n]$, we say {\it a type $i$ agent is in group $G$}, written as $i\in G$, if $i=i_k$ for some $1\le k\le n$. In this case, a type $i$ agent in group $G$ has a {\it group partner} (or a trade partner) of type $i_m$, for all $m\neq k$.

Thirdly, given any assignment $\tau\in \Tau$, we say a group $G\in\mathcal{G}$ is {\it formed}, or is a formed group, under assignment $\tau$ if $G\in \supp(\tau)=\cup_{n=N'}^N \supp(\tau_n)$. In the finite type case, a group is a formed group under assignment $\tau$ if and only if $\tau(G)>0$. i.e. there is a positive mass of group $G$ in the partition represented by assignment $\tau$.

%We note when $N'=N=2$, the group assignment is a matching in the roommate problem. And it is worth noting that there is a probability interpretation of the notion group assignment: in the finite case, for any group $G$ containing i, $\tau(G)/\mu(i)$ is the probability agent i is assigned into the group G 

\subsubsection{Stability}
In the game, players form small groups in order to share group surplus. We say an assignment is stable if there is a way to split group surpluses such that no group of agents can jointly do better by forming a new group. 

Formally, an assignment $\tau\in \Tau$ is {\it stable} if there is an {\it imputation} $u\in L^1(I,\mu)$ satisfying the following two conditions:
\begin{enumerate}
	\item Feasibility: $\sum_{i\in G}u(i)\le s(G)$\footnote{In this paper, $\sum_{i\in G} f(i)=\sum_{i\in I} G(i)f(i)$. This definition is well-defined since $G=0$ at all but finitely many points in its domain. See \hyperref[footnote:infinitesum]{Footnote \ref{footnote:infinitesum}} for a definition of summation over $G$.}, for all formed groups $G\in \supp(\tau)$
	\item No-blocking: $\sum_{i\in G}u(i)\ge s(G)$, for all groups $G\in \mathcal{G}$
\end{enumerate}
In this definition, an imputation specifies a payoff for each type of players. In particular, the same type of player has the same payoff, as otherwise the player with a lower payoff has an incentive to form a group with the group partners of the player with a higher payoff such that all members in the new group have higher payoffs. The feasibility condition ensures that players have small enough payoffs such that they can achieve these payoffs by sharing  group surpluses. The no-blocking condition ensures that players have high enough payoffs such that no group of players can jointly do better by forming a new group.

\subsection{Results}\label{subsection: result1}
In this section, we state our results. The proofs of these results will be given in \hyperref[sect: finitegameproof]{Section \ref{sect: finitegameproof}}. 

The first main theorem in this paper is the existence of a stable assignment when group sizes are finitely bounded.

\begin{theorem}\label{theorem:finiteexistence}
	For any game $((I,\mu),s, N', N)$ satisfying \hyperref[A1]{Assumption (A1)}, there is a stable assignment and its associated imputation is continuous.
%Moreover, when $(A1)'$ is satisfied, the associated imputation is integrable and continuous. 
\end{theorem}

Following the long-term wisdom in works such as \cite{koopmans1957assignment}, \cite{ShapleyShubik71}, \cite{gretsky1992nonatomic}, we prove the existence theorem by establishing a duality relation. In our setting, the duality relation is that the maximum social welfare achieved by forming small groups is equal to the minimum social welfare such that no blocking coalition exists. The formal statement of the duality relation is given by the following theorem.

\begin{theorem} \label{theorem:finiteduality}
	For any game $((I,\mu),s, N', N)$ satisfying \hyperref[A1]{Assumption (A1)}, we  have
	\begin{equation}\label{eqn: duality}
	\sup_{\tau\in \Tau} \sum_{n=N'}^N \int_{\mathcal{G}_n} s_n d\tau_n=\inf_{u\in \mathcal{U}} \int_I u d\mu
	\end{equation}
	where $\mathcal{U}=\{u\in L^1(I,\mu): \sum_{i\in G} u(i)\ge s(G), \forall G\in\mathcal{G}\}$ and the inifimum can be achieved by a continuous function.
\end{theorem}

In particular, this duality theorem suggests that, in a stable state, the payoff of a type $i$ player is equal to his marginal contribution to the maximum social welfare. See \hyperref[subsection: intepretation1]{Section \ref{subsection: intepretation1}} for more details.

%As a corollary of the duality theorem, we have the following uniqueness theorem for the imputation when there is a continuum of types.
%
%\begin{cor} \label{cor: uniqueness}
%	For any game $((I,\mu),s, N', N)$ satisfying \hyperref[A1]{Assumption (A1)}, if we have two more conditions as follows,
%	\begin{enumerate}
%		\item $I$ is a connected and has a local coordinate structure
%		\item $\mu$ has a strictly positive density with respect to the local coordinate
%		%\item for all $s_n$ is continuously differentiable 
%	\end{enumerate}
%	the continuous imputation in the existence theorem (Theorem 1) is unique.
%\end{cor} 
%\begin{proof}
%	By \hyperref[lemma: equivalencerelation]{Lemma \ref{lemma: equivalencerelation}}, it is sufficient to show the continuous minimizer in Theorem 2 is unique. Suppose $u_1, u_2\in C(I)$ minimizes $\int_I u d\mu$ and there is an $i\in I$ such that $u_1(i)<u_2(i)$, by continuity, there is an open neighborhood $S\subset I$ such that $u_1(j)<u_2(j)$ for all $j\in S$. Since $\mu$ has a strictly positive density on $I$, we have $\int_I \min(u_1,u_2) d\mu< \int_I u_2 d\mu=\int_I u_1 d\mu$. But clearly $u $
%\end{proof}

Our key observation in the proof of these two theorems is that any assignment can be identified by a {\it symmetric} transport plan in some replicated agent space. We will discuss this identification trick in more detail in \hyperref[subsubsection:uniquegroupsize]{Section \ref{subsubsection:uniquegroupsize}} and \hyperref[subsubsection: pullback]{Section \ref{subsubsection: pullback}}. As a consequence of the observation, the social welfare maximization problem can be reformulated as a symmetric transport problem. 

\begin{prop}\label{prop:linearprogrammingtotransport}
	For any game $((I,\mu),s, N', N)$ satisfying \hyperref[A1]{Assumption (A1)}, there is a symmetric\footnote{S is symmetric if, for any permutation $\sigma\in S_{N!}$, any $(i_1,...,i_{N!})\in I^{N!}$,  $S(i_1,...,i_{N!})=S(i_{\sigma(1)},...,i_{\sigma(N!)})$.} upper semi-continuous function $S:I^{N!}\rightarrow \mathbb{R}$ such that
	\begin{equation}\label{eqn: lptotransport}
	\sup_{\tau\in \Tau} \sum_{n=N'}^N \int_{\mathcal{G}_n} s_n d\tau_n=\sup_{\hat{\gamma}\in \hat{\Gamma}_{sym}} \int_{I^{N!}} \frac{S}{N} d\hat{\gamma}
	\end{equation}
	where $\hat{\Gamma}_{sym}$ is the set of symmetric measures on $I^{N!}$ such that all marginals are $\mu$. i.e. 
	$$\hat{\Gamma}_{sym}=\{\hat{\gamma}\in \mathcal{M}_+(I^{N!}): \hat{\gamma}\text{~is~symmetric}, \hat{\gamma}(A\times I\times...\times I)=\mu(A),\forall \text{measurable~}A\subset I \}$$
\end{prop}

This proposition helps us in two ways. Conceptually, it helps us to prove the existence of a stable assignment. In particular, it helps us to prove the duality theorem (\hyperref[theorem:finiteduality]{Theorem \ref{theorem:finiteduality}}), which further implies the existence of a stable assignment in the game (\hyperref[theorem:finiteexistence]{Theorem \ref{theorem:finiteexistence}}). Computationally, by generating symmetries to the problem, it helps us to reduce the number of unknowns for finding stable assignments significantly. In particular, this proposition enables the possibility of applying a dimensional reduction technique developed in \cite{friesecke2018breaking}. We will discuss this computational improvement in more detail in a production problem context in \hyperref
[subsubsection: production]{Section \ref{subsubsection: production}}.

In addition, the reformulation has the potential to answer the following questions: 
\begin{enumerate}
	\item the uniqueness of stable assignments
	\item the uniqueness of imputations
	\item whether players of the same type have the same group partners
\end{enumerate}
These questions are related to the uniqueness and purification properties of multi-marginal transport problems. To my knowledge, no existing work can be applied directly beyond the two marginal case. We refer to \cite{pass2011structural} and \cite{pass2015multi} for some related results. 

\subsubsection{Remarks}\label{subsection: intepretation1}

The duality relation in \hyperref[theorem:finiteduality]{Theorem \ref{theorem:finiteduality}} suggests that the payoff of a type $i$ player in any stable assignment is equal to his marginal contribution to the maximum social welfare. For instance, in the finite type case, the maximum social welfare, as a function of type distribution, is defined by a function $\Pi:\mathbb{R}^{|I|}_{+}\rightarrow \mathbb{R}$, where
\begin{equation} \label{eqn: finitemax}
\Pi(\mu)=\sup_{\tau\in \Tau} \sum_{G\in\mathcal{G}} s(G)\tau(G)
\end{equation}

Since there are a continuum of players of each type, the maximum social surplus function $\Pi$ is a concave function on $\mathbb{R}_+^{|I|}$.\footnote{For any type distributions $\mu_1$ and $\mu_2$, we fix two arbitrary assignments $\tau_1$ and $\tau_2$ representing the partitions of players with a type distribution $\mu_1$ and $\mu_2$ respectively. Any convex combination $c\tau_1+(1-c)\tau_2$ is an assignment representing a partition of players with a type distribution $c\mu_1+(1-c)\mu_2$.} Therefore, the super-derivative of the maximum social surplus is well-defined. 

On the other hand, the duality theorem, \hyperref[theorem:finiteduality]{Theorem \ref{theorem:finiteduality}}, suggests,
\begin{equation}\label{eqn: finitemin}
\Pi(\mu)=\inf_{u\in \mathcal{U}} \sum_{i\in I} u(i)\mu(i)
\end{equation}
where $\mathcal{U}=\{ u\in \mathbb{R}_+^{|I|}: \sum_{i\in G} u(i) \ge s(G),\forall G\in \mathcal{G}  \}$. Therefore, the imputation $u$ associated with a stable assignment is equal to the super-derivative of the concave maximum social welfare function. Formally, by Danskin's theorem (Proposition B.25 in \cite{bertsekas1999nonlinear}), 
$$\partial\Pi(\mu)=\left\{u\in \mathcal{U}: \sum_{i\in I} u(i)\mu(i) = \Pi(\mu) \right\}$$
That is, the payoff of a type $i$ player is equal to his marginal contrition to the maximum social surplus, provided the maximum social surplus function is differentiable. More generally, the imputations provide separating hyperplanes for the set of feasible social surplus.

\begin{figure}[h]
	\centering
	\includegraphics[width=0.8\textwidth]{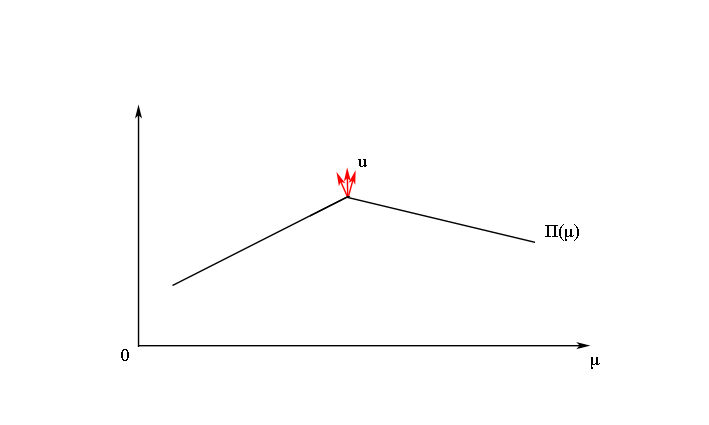}
	\caption{Maximum Social Welfare and The Imputation of a Stable Assignment}
	\label{fig: nondiff}
\end{figure}

By \hyperref[fig: nondiff]{Figure \ref{fig: nondiff}}, it is clear that, given a stable assignment, the corresponding imputation is unique if and only if the maximum social surplus is differentiable. Unfortunately, the differentiability requires stronger assumptions on the surplus function. An example of a game with a non-differentiable maximum social surplus function is given in \hyperref[subsubsection: nondiff]{Section \ref{subsubsection: nondiff}}. On the other hand, as the maximum social welfare function $\Pi$ is concave, $\Pi$ is differentiable almost everywhere. That is, in the finite type case, for almost all initial distributions of types, the payoff at the stable state is unique.

\subsubsection{Extensions} \label{subsection: remarks}

As remarked in \hyperref[footnote: cptassumption]{Footnote \ref{footnote: cptassumption}} and \hyperref[footnote: ctsassumption]{Footnote \ref{footnote: ctsassumption}}, our model and results in this section can be extended to the case where the type set is not compact and the surplus function is only upper semi-continuous. In this subsubsection,  we discuss two types of extensions.

Firstly, when $I$ is a (possibly non-compact) Polish space and the surplus function is Lipchitz continuous, i.e. every $s_n$ is Lipchitz continuous, 
with all definitions in this section remain the same, all results in this section hold, provided the surplus function satisfies Assumption $(A2^c)$: for any permitted group size $n$, there is a function $a_{n}\in C(I)$ such that 		$s_n([i_1,...,i_n])\le \sum_{k=1}^n a_{n}(i_k)$ for all $i_1,...,i_n\in I$. In particular, Assumption $(A2^c)$ is automatically satisfied when $s_n$ is bounded or $I$ is compact.

More generally, if $I$ is a (possibly non-compact) Polish space and the surplus function is upper semi-continuous, i.e. every $s_n$ is upper semi-continuous, we need to impose the following boundedness assumption on the surplus function:
\begin{itemize}
	\item [(A2)]\label{A2} for any permitted group size $n$, there is a function $a_{n}\in LSC(I)$ such that 	
	$s_n([i_1,...,i_n])\le \sum_{k=1}^n a_{n}(i_k)$ for all $i_1,...,i_n\in I$. 
\end{itemize} 
Again, Assumption $(A2)$ is automatically satisfied when $s_n$ is bounded or $I$ is compact. With Assumption $(A2)$ and the following changes on the definition of ``formed groups", all results in this section hold:
\begin{enumerate}
	\item in the definition of stable assignment, replace the term ``for all formed groups $G\in\supp(\tau)$" in the feasibility condition by the term ``for $\tau$-almost all $G\in\supp(\tau)$"\footnote{\label{footnote: redefformedgroups}Equivalently, we can redefine a notion for formed groups in order to keep the term ``for all formed groups" in the sentence. The definition is as follows: a measurable set $\mathcal{F}\subset \supp(\tau)$ is  a set of formed groups if $\mathcal{F}$ has a full $\tau$-measure. Formally, it means $\tau_n(\mathcal{F}\cap \mathcal{G}_n)=\tau_n(\mathcal{G}_n)$ for all permitted group size $n$. Then, in the definition of stable assignment, we just need to replace ``for all formed groups $G\in\supp(\tau)$" by ``for all formed groups $G\in \mathcal{F}$ where $\mathcal{F}$ is a set of formed groups". This subtlety about measure zero set comes from the measure theoretic language we choose to use.}
	\item \hyperref[theorem:finiteexistence]{Theorem \ref{theorem:finiteexistence}} holds if ``and its associated imputation is continuous" is removed.
	\item \hyperref[theorem:finiteduality]{Theorem \ref{theorem:finiteduality}} holds if ``by a continuous function" is removed.
	\item \hyperref[prop:linearprogrammingtotransport]{Proposition \ref{prop:linearprogrammingtotransport}} holds without any change.
\end{enumerate}

%$$s_n([i_1,...,i_n])\le \sum_{k=1}^n a_{n}(i_k)$$

%
%Since, in the finite case, continuity and integrability are hold trivially, the existence of the stable assignment needs no essential assumption when the type space is finite:
%\begin{cor1}
%	For the cooperative game with limit group structure $((I,\mu),s, N', N)$ where $I$ is a finite set, the set of stable assignment is non-empty.
%\end{cor1}

\subsection{Proof}\label{sect: finitegameproof}
Firstly, we prove \hyperref[prop:linearprogrammingtotransport]{Proposition \ref{prop:linearprogrammingtotransport}} which states the social welfare maximization problem can be reformulated as a symmetric transport problem. Then, we apply the duality theorem for multi-transport problem in \cite{Kellerer84} to prove the duality theorem \hyperref[theorem:finiteduality]{Theorem \ref{theorem:finiteduality}}. Lastly, we show the optimizers of the duality theorem can be attained and are equivalent to a stable assignment and its corresponding imputation.

To start with, we reformulate the social welfare maximization problem as a symmetric transport problem in two steps. 

In the first step, we transform a problem with multiple permitted group sizes up to $N$ to a problem with a unique permitted group size $N!$. We use the number $N!$ as it is a common multiple of all group sizes. The same argument works if we replace $N!$ by any common multiple of group sizes\footnote{This replacement will help the computation significantly. However, it makes the notations messier. So we stick with the choice $N!$ in this paper.}.

In the second step, we extend the domain of assignment from an unordered list of types to an ordered list of types. On the technical level, we get rid of equivalent classes by defining a pull-back measure carefully.

The proof is proceeded as follows. Firstly, we introduce these two steps formally in \hyperref[subsubsection:uniquegroupsize]{Section \ref{subsubsection:uniquegroupsize}} and \hyperref[subsubsection: pullback]{Section \ref{subsubsection: pullback}}. Next, we reformulate the social welfare maximization problem as a symmetric transport problem in \hyperref[subsubsection: maximizationastranport]{Section \ref{subsubsection: maximizationastranport}}. Lastly, we prove the existence of a stable assignment in \hyperref[subsubsection: existenceproof]{Section \ref{subsubsection: existenceproof}}.

\subsubsection{Unifying Group Sizes}\label{subsubsection:uniquegroupsize}

To unify the group sizes, we use fractional groups to identify groups of different sizes.

Formally, a fractional group is an element $\hat{G}=[i_1,...,i_{N!}]\in I^{N!}/\sim_{N!}$. Intuitively, a fractional group $\hat{G}=[i_1,...,i_{N!}]$ represents a set of players of total mass $N$ consisting of  $\frac{1}{(N-1)!}$ unit mass of type $i_n$ players for all permitted group size $n$. 

For any permitted group size $n$, an $n$-person group can be identified by some fractional group via replication. Therefore, the set of $n$-person groups corresponds to a collection of fractional groups. Formally, for any permitted group size $N'\le n\le N$, we define a subset $K_n\subset I^{N!}/\sim_{N!}$, where
\begin{equation}\label{eqn: Kn}
K_n=\left\{[i_1,...,i_{N!}]\in I^{N!}/\sim_{N!}:   |k\in \mathbb{N}: i_k=i| ~\text{is~divisible~by~}\frac{N!}{n}, \forall i\in I  \right\} 
\end{equation}
Intuitively, every fractional group in $K_n$ corresponds to an $N!/n$-fold replication of some $n$-person group. We define an identification map $P_n: K_n\rightarrow I^n/\sim_n$ by
$$P_n([\underbrace{i_1,...,i_1}_{N!/n \text{~many}},...,\underbrace{i_n,...,i_n}_{N!/n \text{~many}}])=[i_1,...,i_n]$$
It is easy to verify that the identification map $P_n$ is bijective. Therefore, $\mathcal{G}_n\simeq K_n$ and $\mathcal{G}_n=P_n(K_n)$.

In addition, we extend the domain of surplus function: for any $N'\le n\le N$, we define a function $\hat{s}_n:I^{N!}/\sim_{N!}\rightarrow \mathbb{R}$ by
$$\hat{s}_n(\hat{G})=\begin{cases}
s_n(P_n(\hat{G})), & \hat{G}\in K_n \\
0, & \hat{G}\notin K_n
\end{cases}$$  

Next, we define a surplus function on the set of fractional groups. We note that the collection of subsets $\{K_n\}_{N'\le n\le N}$ is not a partition of the set of fractional groups $I^{N!}/\sim_{N!}$. Therefore, we cannot combine the functions $\hat{s}_n$ to define this surplus function directly. However, we note that a welfare maximizing assignment assigns positive masses to two groups, which correspond to the same fractional group, if and only if the average surplus of these two groups are the same\footnote{Otherwise, only the group with a higher average surplus will be formed in any welfare maximized assignment.}.  
Therefore, we define the surplus function 
$\hat{s}:I^{N!}/\sim_{N!}\rightarrow \mathbb{R}$ by
\begin{equation}\label{eqn: hats}
\hat{s}=
\max\left(\frac{N}{n} \hat{s}_n: n\in \{N',...,N\}\right) \footnotemark
\end{equation}
\footnotetext{In the numerator, it is $N$ rather than $N!$ as total mass of players in a fractional group is normalized to be $N$.}
The surplus function $\hat{s}$ on fractional groups induces a partition $\{R_n\}_{n=0,N',...,N}$ of the space $I^{N!}/\sim_{N!}$ such that
\begin{itemize}
	\item $R_0\cap K_n=\emptyset, \forall n\in \{N',...,N\}$
	\item $R_n\subset K_n$ and $\hat{s}_n\ge \frac{n}{m} \hat{s}_m \text{~on~} R_n, \forall N'\le n,m\le N$
\end{itemize}
Such partition $\{R_n\}_{n=0,N',...,N}$ can be constructed iteratively:
	\begin{itemize}
		\item define $R_0 = I^{N!}/\sim_{N!} - \bigcup_{n=N'}^N K_n$, and $R_n=\emptyset$ for all $n\in\{N',...,N\}$
		\item set $n=N'$
		\item update $R_n$ to be $\{\hat{G}\in K_n-\cup_{k=n}^{N} R_k:\hat{s}(\hat{G})=\frac{N}{n} \hat{s}_n(\hat{G})\}$ for all $n\in\{N',...,N\}$
		\item if $n<N$, set $n$ to be $n+1$ and repeat the previous step, otherwise stops
\end{itemize}

%
%{\centering
%\begin{minipage}{.8\linewidth}
%\begin{algorithm}[H]
%	\caption{Construction of the partition $\{R_n\}_{n=0,N',...,N}$}\label{algorithm: partitionconstruction}
%	\begin{algorithmic}[H]
%		\Procedure{Partition}{$\{K_n,\hat{s}_n\}_{n\in \{N',...,n\}}$}
%		\State $R_0 \gets I^{N!}/\sim_{N!} - \bigcup_{n=N'}^N K_n$
%		\State $R_n \gets \emptyset, \forall n\in\{N',...,N\}$
%		\For{$n\in \{N,...,N'\}$}
%		\State $R_n\gets \{\hat{G}\in K_n-\cup_{k=n}^{N} R_k:\hat{s}(\hat{G})=\frac{N}{n} \hat{s}_n(\hat{G})\}$
%		%\State $s\gets \max(s,\frac{N}{n}\hat{s}_n)$ \Comment{$\hat{s}_n=0$ outside the domain}
%		\EndFor
%		\State \textbf{return} $\{R_n\}_{n=0,N',...,N}$%\Comment{The gcd is b}
%		\EndProcedure			
%	\end{algorithmic}
%\end{algorithm}
%\end{minipage}
%\par}
%~\\

Moreover, the surplus function $\hat{s}$ on fractional groups inherits the properties of the surplus function $s$:

\begin{lemma}\label{lemma:shatprop}
	If the surplus function $s$ satisfies \hyperref[A1]{Assumption (A1)} and \hyperref[A2]{Assumption (A2)}, then $\hat{s}$ is upper semi-continuous and there is a lower semi-continuous function $\hat{a}\in L^1(I,\mu)$ such that 
	$$\hat{s}([i_1,...,i_{N!}])\le \sum_{k=1}^{N!} \hat{a}(i_k)$$
\end{lemma}

\begin{proof}
	See \hyperref[appendix:shatprop]{Appendix \ref{appendix:shatprop}}.
\end{proof}

\subsubsection{A Change of Variable Trick}\label{subsubsection: pullback}
Next, we extend the domain of a measure. In particular, a non-negative measure $\tau$ on $I^n/\sim_{n}$ will be identified by a symmetric measure $\gamma$ on $I^n$. Here, $n$ is a natural number larger than or equal to 2. 

Formally, given any non-negative measure $\tau_n$ on $I^n/\sim_n$, we define a non-negative measure $\gamma_n$ on $I^n$ by
\begin{equation}\label{eqn: defgamma}
\gamma_n=c_n Q_n^{\#}\tau_n
\end{equation}
where $c_n:I^n\rightarrow\mathbb{R}$ is a measurable  function\footnote{See  \hyperref[appendix:cmeasurability]{Appendix \ref{appendix:cmeasurability}} for the proof of measurability.} defined by the combinatorial numbers
\begin{equation}\label{eqn:defc}
c_n(i_1,...,i_n)=\frac{1}{(n-1)!}\prod_{i\in I} n_i!~~ \footnote{There are finitely many $i\in I$ such that $n_i\neq 0$. The product over an infinite set $\prod_{i\in I} n_i!$ is defined to be $\prod_{i\in I: n_i\neq 0} n_i!$}
\end{equation}
where $n_i=|\{k\in \mathbb{N}: i_k=i\}|$ is the number of  type $i$ players in group $[i_1,...,i_n]\in\mathcal{G}_n$,
and $Q_n:I^n\rightarrow I^n/\sim_n$ is the quotient map such that
$$Q_n(i_1,...,i_n)=[i_1,...,i_n]$$ 

In order to define the pullback measure $Q_n^{\#}\tau_n$ on $I^n$, we partition $I^n$ into small components such that $Q_n$ is an injective map on each component.

\begin{lemma} \label{lemma: injectiveparition}
There is a partition $\{J_{\alpha}\}_{\alpha\in \mathrm{A}}$ of $I^n$ such that
\begin{itemize}
	\item the index set $\mathrm{A}$ is a finite set
	\item each component $J_\alpha$ is Borel measurable in $I^n$
	\item the quotient map $Q_n$ is injective on $J_\alpha$
\end{itemize}
\end{lemma}
\begin{proof}
We partition $I^n$ in two steps.

Firstly, we partition $I^n$ into $|\mathcal{M}|$ components, where the index set $\mathcal{M}$ is defined by 
$$\mathcal{M}=\bigcup_{k=1}^n \{m=(m_1,...,m_k)\in \mathbb{N}^k: m_1\ge...\ge m_k\ge 1, m_1+...+m_k=n\}$$
A component $J_m$ with a index $m=(m_1,...,m_k)\in \mathcal{M}$ consists of all elements in the set
$$\{(\underbrace{i_1,...,i_1}_{m_1~\text{many}},\underbrace{i_2,...,i_2}_{m_2~\text{many}},...,\underbrace{i_k,...,i_k}_{m_k~\text{many}})\in I^n: i_1,...,i_k \text{~are~disjoint}\}$$
and all their permutations. We know $J_m$ is a measurable set\footnote{See \hyperref[appendix:cmeasurability]{Appendix \ref{appendix:cmeasurability}} for a proof.}.

Secondly, for each index $m\in\mathcal{M}$, we partition $J_m$ into $\frac{n!}{m_1!...m_k!}$ components such that $Q_n$ is injective on each component. The construction is as follows: for each permutation $\sigma\in S_n$, we define
$$J_{m,\sigma}=\{(j_{\sigma(1)},...,j_{\sigma(n)})\in I^n: j=(\underbrace{i_1,...,i_1}_{m_1~\text{many}},\underbrace{i_2,...,i_2}_{m_2~\text{many}},...,\underbrace{i_k,...,i_k}_{m_k~\text{many}})\in I^n, i_1,...,i_k\text{~are~disjoint}\}$$
$J_{m,\sigma}$ is a measurable set\footnote{See \hyperref[appendix:cmeasurability]{Appendix \ref{appendix:cmeasurability}} for a proof.}. Moreover, for any permutations $\sigma_1,\sigma_2\in S_n$, we have either $J_{m,\sigma_1}=J_{m,\sigma_2}$ or $J_{m,\sigma_1}\cap J_{m,\sigma_2}=\emptyset$. Moreover, $\cup_{\sigma\in S_n} J_{m,\sigma}=J_m$. By deleting the repeated components from $\{J_{m,\sigma}\}_{\sigma\in S_n}$, we have a partition of the set $J_m$ such that $Q_n$ is injective on each component.  

In sum, by putting the indices $(m,\sigma)$ together as a single index set $\mathrm{A}$, we have the partition.
\end{proof}

With the partition $(J_\alpha)_{\alpha\in \mathrm{A}}$ of $I^n$, we define the pullback measure $Q_n^{\#}\tau_n$ on $I^n$ by
\begin{equation}\label{eqn: pullback}
Q_n^{\#}\tau_n(S)=\sum_{\alpha\in\mathrm{A}}\tau_n(Q_n(S\cap J_\alpha)), \forall \text{measurable~}  S\subset I^n
\end{equation}

We verify our definition of pullback measure by proving that $\tau_n$ is the push-forward measure of $\gamma_n$ under the map $Q_n/n$:
\begin{lemma} \label{lemma:gamma}
If $\gamma_n$ is defined by $\tau_n$ according to \hyperref[eqn: defgamma]{Equation \ref{eqn: defgamma}} and \hyperref[eqn: pullback]{Equation \ref{eqn: pullback}}, we have 
\begin{equation}\label{eqn: gammarelation}
n\tau_n=(Q_n)_\#\gamma_n
\end{equation}
\end{lemma}
\begin{proof}
	See \hyperref[appendix: gamma]{Appendix \ref{appendix: gamma}}.
\end{proof}

%
%\begin{lemma} \label{Lemma: gamma}
%For any measurable function $g:I^n/\sim_n \rightarrow \mathbb{R}$, 
%\begin{equation}\label{eqn:defc1}
%\int_{I^n} g\circ Q_n d\gamma_n=n\int_{I^n/\sim_n} g d\tau_n 
%\end{equation}
%%In particular, for any measurable subset $S\subset I^n$,
%%\begin{equation} \label{eqn: defc2}
%%\gamma_n(S)=\int_{I^n/\sim_n} c(Q_n^{-1}(x))\chi_{Q_n(S)}(x)d\tau_n=\int_{Q_n(S)} c\circ Q_n^{-1} d\tau_n
%%\end{equation}
%\end{lemma}

In the finite type case, \hyperref[lemma:gamma]{Lemma \ref{lemma:gamma}} implies, for any list of types $(i_1,...,i_n)\in I^n$,
$$n\tau_n([i_1,...,i_n])=\gamma_n(Q_n^{-1}([i_1,...,i_n]))=\sum_{i_1,...,i_n\in I} \gamma_n(i_1,...,i_n)$$
That is, for all groups $G=[i_1,...,i_n]$, the sum of the $\gamma$ values over all permutations of the group $G$ is equal to the total mass of players assigned to group $G$. Since $\gamma_n$ is a symmetric function, we have
$$\gamma_n(i_1,...,i_n)=c_n(i_1,...,i_n)\tau_n([i_1,...,i_n])$$
which gives us \hyperref[eqn: defgamma]{Equation \ref{eqn: defgamma}}. In particular, $\gamma_n(i_1,...,i_n)$ is the product of some positive constant and the mass of the $n$-person group $[i_1,...,i_n]$ in assignment $\tau\in \Tau$. Recall that a group $G=[i_1,...,i_n]$ is formed when $\tau([i_1,...,i_n])>0$. Therefore, an $n$-person group $G\in \mathcal{G}_n$ is formed if and only if $\gamma_n(Q_n^{-1}(G))>0$ for some permitted group size  $n$. That is, a preimage $Q_n^{-1}(G)$ is in the support of $\gamma_n$ for some permitted group size $n$. Consequently, the positive constant $c$ plays no role in our definition of stability. 

Other properties of $\gamma_n$ is given as follows: 

\begin{lemma}\label{lemma:gammaprop}
For any measure $\gamma_n$ defined by some measure $\tau_n$ and \hyperref[eqn: gammarelation]{Equation \ref{eqn: gammarelation}}, we have
\begin{enumerate}
	\item $\gamma_n$ is symmetric: for any measurable sets $A_1,...,A_n\subset I$, $$\gamma_n(A_1,...,A_n)=\gamma_n(A_{\sigma(1)},...,A_{\sigma(n)}), \forall \sigma\in S_n$$
	\item $\gamma_n(A,I,I...,I)=\sum_{k=0}^n k \tau_n(\mathcal{G}_n(A,k))$ for all measurable $A\subset I$
	%\item $\sum_{n=N'}^{N}\gamma_n(A,I,I...,I)=\mu(A)$ for all measurable $A\subset I$
\end{enumerate}
\end{lemma}
\begin{proof}
	See \hyperref[appendix:gammaprop]{Appendix \ref{appendix:gammaprop}}.
\end{proof}

An immediate consequence of \hyperref[lemma:gammaprop]{Lemma \ref{lemma:gammaprop}} 
is that, when there is a unique group size, i.e. $N'=N$, $\gamma_N$, defined by an assignment $\tau_N\in \Tau$, is a symmetric measure such that all its marginals are $\mu$.
That is, $\gamma_N$ is a symmetric transport plan in the $N$-fold replicated player space.

Lastly, we illustrate the change of variable trick
in a game with three types of agents. That is, $I=\{i_1,i_2,i_3\}$. In this game, the unique permitted group size is 2.

Recall that an assignment $\tau\in \Tau$ is a function defined on $\mathcal{G}=\mathcal{G}_2=I^2/\sim_2$ satisfying the consistency condition
$$\sum_{m\neq n}\tau_{nm}+2\tau_{nn}=1, \forall n\in \{1,2,3\}$$
where $\tau_{nm}=\tau([i_n,i_m])$ for any $n\le m$. In particular, $\tau_{nm}$ is the mass of groups consisting of one type $i_n$ player and one type $i_m$ player in assignment $\tau$. Therefore, an assignment $\tau$ can be represented by the upper triangular entries of a 3 by 3 matrix
$$\tau=\begin{pmatrix} 
\tau_{11} & \tau_{12} & \tau_{13}  \\
 &  \tau_{22} & \tau_{23} \\
 &  & \tau_{33}\\
\end{pmatrix}$$

Naturally, we wish to extend the domain of $\tau$ to the full domain $I^2$ such that an assignment can be presented by a matrix. There are infinitely many ways to extend the domain. Among them, one of the most natural ways is to extend the upper triangular entries symmetrically:
$$Q_3^{\#}\tau=\begin{pmatrix} 
\tau_{11} & \tau_{12} & \tau_{13}  \\
\tau_{21}&  \tau_{22} & \tau_{23} \\
\tau_{31}&  \tau_{32}& \tau_{33}\\
\end{pmatrix}$$

While $Q_3^{\#}\tau$ is a symmetric matrix, it is not a transport plan as it does not satisfy the marginal condition. In this case, it means the matrix is not bi-stochastic. To obtain a bi-stochastic matrix, we multiply the diagonal entries by 2:
$$\gamma=cQ_3^{\#}\tau=\begin{pmatrix} 
2\tau_{11} & \tau_{12} & \tau_{13}  \\
\tau_{21}&  2\tau_{22} & \tau_{23} \\
\tau_{31}&  \tau_{32}& 2\tau_{33}\\
\end{pmatrix}$$
By the consistency condition of the assignment $\tau$, $\gamma$ is bi-stochastic. Therefore, we obtain a symmetric transport plan $\gamma$ in a two-folds replication of agent space.%\footnote{In the discrete players setting, \cite{chiappori2014roommate} studies a version of the inverse process of our identification trick. That is, any pure matching in the replicated agent space corresponds to an assignment. In the continuum players setting we consider, our identification trick helps in two ways. Firstly,  it helps us to study games with a continuum of players and small groups of sizes larger than 2. Secondly, it ensures all stable assignments can be found by solving a symmetric optimal transport problem.}. 
\begin{figure}[h]
	\centering	
	\includegraphics[width=0.35\textwidth]{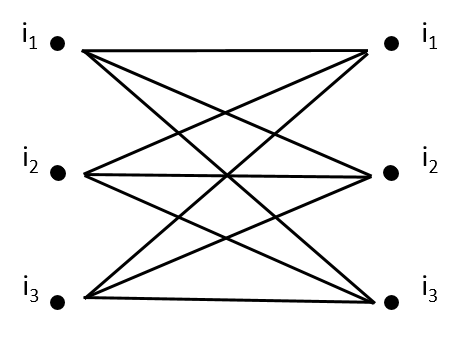}
	\caption{The symmetric transport plan $\gamma$ in a two-folds replicated agent space}
	\label{fig: symmetrictransport}
\end{figure}

A symmetric transport plan $\gamma$ is represented in \hyperref[fig: symmetrictransport]{Figure \ref{fig: symmetrictransport}}. For any $s,t\in\{1,2,3\}$, the mass transported from a source $i_s$ on the left to a destination $i_t$ on the right is $\gamma_{st}$. By the symmetry of the transport plan $\gamma$, the mass transported from a source $i_s$ on the left to a destination $i_t$ on the right is equal to the mass transported from a source $i_t$ on the left to a destination $i_s$ on the right.
	
%\begin{figure}[h]
%	\centering	
%	\includegraphics[width=0.9\textwidth]{transportandsymmetrictransport.png}
%	\caption{an asymmetric transport plan  (left) and a symmetric transport plan (right)}
%	\label{fig: symmetrictransport}
%\end{figure}

\subsubsection{Symmetric Transport Problem}
\label{subsubsection: maximizationastranport}

Next, we prove \hyperref[prop:linearprogrammingtotransport]{Proposition \ref{prop:linearprogrammingtotransport}} by using the two tricks introduced in the previous two subsubsections.

Firstly, we extend the domain of surplus function on fractional groups to the whole domain by defining $S:I^{N!}\rightarrow \mathbb{R}$ where
\begin{equation}\label{eqn: bigs}
S(i_1,...,i_{N!})=\hat{s}([i_1,...,i_{N!}])
\end{equation}
where $\hat{s}$ is defined in \hyperref[eqn: hats]{Equation \ref{eqn: hats}}. By \hyperref[lemma:shatprop]{Lemma \ref{lemma:shatprop}}, when the surplus function satisfies assumption \hyperref[A1]{Assumption (A1)} and \hyperref[A2]{Assumption (A2)}, $S$ is upper semi-continuous in $I^{N!}$ and there is a lower semi continuous $\hat{a}\in L^1(I,\mu)$ such that
\begin{equation}\label{eqn: bigSbdd}
S(i_1,...,i_{N!})\le \sum_{k=1}^{N!} \hat{a}(i_k), \forall i_1,...,i_{N!}\in I
\end{equation}

Now, we prove \hyperref[prop:linearprogrammingtotransport]{Proposition \ref{prop:linearprogrammingtotransport}} by proving \hyperref[eqn: lptotransport]{Equation \ref{eqn: lptotransport}} holds. i.e.
	$$\sup_{\tau\in\Tau} \sum_{n=N'}^N \int_{\mathcal{G}_n} s_n d\tau_n=\sup_{\hat{\gamma}\in \hat{\Gamma}_{sym}} \int_{I^{N!}} \frac{S}{N} d\hat{\gamma}$$

Firstly, we show the left hand side is not larger than the right hand side in \hyperref[eqn: lptotransport]{Equation \ref{eqn: lptotransport}}: for any fixed $\tau\in \Tau$, we construct $\hat{\gamma}\in\hat{\Gamma}_{sym}$ such that
$$\sum_{n=N'}^N \int_{\mathcal{G}_n} s_n d\tau_n\le \int_{I^{N!}}\frac{S}{N} d\hat{\gamma}$$
Given an assignment $\tau\in \Tau$, we define a measure $\hat{\tau}_n$ on $I^{N!}/\sim_{N!}$ such that for any $S\subset I^{N!}/\sim_{N!}$,
$$\hat{\tau}_n(S)=\tau_n (P_{n} (S\cap K_n))$$
where $K_n$ is defined in \hyperref[eqn: Kn]{Equation \ref{eqn: Kn}}. Applying the change of variable trick in \hyperref[subsubsection: pullback]{Section \ref{subsubsection: pullback}}, we have a measure $\hat{\gamma}_n$ on $I^{N!}$ where
$$\hat{\gamma}_n=c\cdot Q_{N!}^{\#} \hat{\tau}_n$$
where $c:I^{N!}\rightarrow \mathbb{R}$ is a measurable function defined by the combinatorial numbers 
$$c(i_1,...,i_{N!})=\frac{1}{(N!-1)!}\prod_{i\in I}n_i!$$
where $n_i=|\{k\in\mathbb{N}: i_k=i\}|$. By \hyperref[lemma:gammaprop]{Lemma \ref{lemma:gammaprop}},
\begin{align*}
\hat{\gamma}_n(A\times I\times ... \times I)=& \sum_{k=0}^{N!} k\hat{\tau}_n(\mathcal{G}_{N!}(A,k))\\
=&\sum_{k=0}^{n} \frac{N!k}{n}\hat{\tau}_n(\mathcal{G}_{N!}(A, \frac{N!k}{n}))\\
=&\frac{N!}{n}\sum_{k=0}^{n} k\tau_n(\mathcal{G}_n(A,k))
\end{align*}
Therefore, we define
$$\hat{\gamma}=\frac{1}{N!}\sum_{n=N'}^N n \hat{\gamma}_n$$
By \hyperref[lemma:gammaprop]{Lemma \ref{lemma:gammaprop}}, $\hat{\gamma}$ is symmetric. By \hyperref[eqn: consistency]{Equation \ref{eqn: consistency}},  $\hat{\gamma}_n(A\times I\times ... \times I)=\mu(A)$ for any Borel measurable $A\subset I$. Consequently, $\hat{\gamma}\in \hat{\Gamma}_{sym}$.

Moreover, by \hyperref[eqn: hats]{Equation \ref{eqn: hats}},
$$\sum_{n=N'}^N \int_{\mathcal{G}_n} s_n d\tau_n =  \sum_{n=N'}^N \int_{I^{N!}/\sim_{N!}} \hat{s}_n d\hat{\tau}_n
\le  \sum_{n=N'}^N n\int_{I^{N!}/\sim_{N!}} \frac{\hat{s}}{N} d\hat{\tau}_n$$
Applying \hyperref[lemma:gamma]{Lemma \ref{lemma:gamma}}, we have
$$\int_{I^{N!}/\sim_{N!}} \frac{\hat{s}}{N} d\hat{\tau}_n=\frac{1}{N!}\int_{I^{N!}} \frac{S}{N} d\hat{\gamma}_n$$
Therefore,
$$\sum_{n=N'}^N \int_{\mathcal{G}_n} s_n d\tau_n\le \sum_{n=N'}^N \frac{n}{N!}\int_{I^{N!}} \frac{S}{N} d\hat{\gamma}_n=\int_{I^{N!}} \frac{S}{N} d\hat{\gamma}$$

Conversely, we show the left hand side is not less than the right hand side in \hyperref[eqn: lptotransport]{Equation \ref{eqn: lptotransport}}: for any maximizer $\hat{\gamma}\in \hat{\Gamma}_{sym}$ of the right hand side of \hyperref[eqn: lptotransport]{Equation \ref{eqn: lptotransport}}\footnote{The existence of a maximizer is proved in the next subsubsection.}, we construct a $\tau\in \Tau$ such that
\begin{equation}\label{eqn: constructtau}
\sum_{n=N'}^N \int_{\mathcal{G}_n} s_n d\tau_n= \int_{I^{N!}}\frac{S}{N} d\hat{\gamma}
\end{equation}

Firstly, since $\hat{\gamma}$ is a maximizer, the support of $\hat{\gamma}$ cannot intersect the set $Q_{N!}^{-1}(R_0)$. Otherwise, a normalized version of $\hat{\gamma}|_{Q_{N!}^{-1}(R_0)^c}$ will give a higher value for the maximization problem.

%Define $\hat{\tau}$ on the set of fractional groups $I^{N!}/\sim_{N!}$ by
%$$\hat{\tau}=\frac{1}{N!}(Q_{N!})_{\#}\hat{\gamma}$$
%By definition,
%$$\int_{I^{N!}}\frac{S}{N} d\hat{\gamma}= (N-1)!\int_{I^{N!}/\sim_{N!}}\hat{s} d\hat{\tau}$$

For every permitted group size $n\in \{N',...,N\}$, we define a measure $\hat{\gamma}_n$ on $I^{N!}$ by 
$$\hat{\gamma}_n(S)=\frac{N!}{n}\hat{\gamma}(S\cap Q_{N!}^{-1}(R_n)), \forall \text{measurable~} S\subset I^{N!}$$ 
In particular, the support of $\gamma_n$ lies in the set $Q_{N!}^{-1}(R_n)$. By the definition of $\hat{\Gamma}_{sym}$ and the fact that the support of $\hat{\gamma}$ is in $\cup_{n=N'}^N Q_{N!}^{-1}(R_n) $, we have
\begin{align}
\mu(A)&=\hat{\gamma}(A\times I\times...\times I) \nonumber \\
&=\hat{\gamma}(A\times I\times...\times I\cap (\cup_{n=N'}^N Q_{N!}^{-1}(R_n) )) \nonumber \\
&= \sum_{n=N'}^N \hat{\gamma}(A\times I\times...\times I\cap  Q_{N!}^{-1}(R_n) )\nonumber \\
&=\sum_{n=N'}^N \frac{n}{N!}\hat{\gamma}_n(A\times I\times...\times I) \label{eqn: gammahatmarginal}
\end{align}

Next, for every permitted group size $n$, we define a measure
$\hat{\tau}_n$ on the set of fractional groups $I^{N!}/\sim_{N!}$ by
$$\hat{\tau}_n=(Q_{N!})_{\#}\hat{\gamma}_n$$
%By definition,
%$$\int_{I^{N!}}\frac{S}{N} d\hat{\gamma}= (N-1)!\int_{I^{N!}/\sim_{N!}}\hat{s} d\hat{\tau}$$

Applying \hyperref[lemma:gammaprop]{Lemma \ref{lemma:gammaprop}} to \hyperref[eqn: gammahatmarginal]{Equation \ref{eqn: gammahatmarginal}}, we have
$$\mu(A)= \sum_{n=N'}^N \frac{n}{N!}\hat{\gamma}_n(A\times I\times...\times I)=\sum_{n=N'}^N  \sum_{k=0}^{N!} k \frac{n}{N!}\hat{\tau}_n(\mathcal{G}_{N!} (A,k) )$$

Recall the identification map $P_n:K_n\subset I^{N!}/\sim_{N!}\rightarrow \mathcal{G}_n$, we define a measure $\tau_n$ on the set of $n$-person groups $\mathcal{G}_n\simeq I^n/\sim_n$ by
$$\tau_{n}(S)=\hat{\tau}_n (P_n^{-1}(S)), \forall\text{ measurable~} S\subset \mathcal{G}_n$$ 
Therefore,
$$\sum_{k=0}^{N!} k \hat{\tau}_n(\mathcal{G}_{N!} (A,k) )=\sum_{k=0}^{n} \frac{N!}{n}k \hat{\tau}_n(\mathcal{G}_{N!} (A,\frac{N!}{n}k))=\sum_{k=0}^{n} \frac{N!}{n}k \tau_n(\mathcal{G}_{n} (A,k))$$
Consequently, we have
$$\mu(A)=\sum_{n=N'}^N\sum_{k=0}^{n} k \tau_n(\mathcal{G}_{n} (A,k))$$
i.e. $\tau=(\tau_{N'},...,\tau_N)\in\Tau$ is an assignment.

Lastly, we check that \hyperref[eqn: constructtau]{Equation \ref{eqn: constructtau}} holds:

\begin{align*}
\int_{I^{N!}}\frac{S}{N} d\hat{\gamma}=&\sum_{n=N'}^N \int_{Q_{N!}^{-1}(R_n)}\frac{S}{N} d\hat{\gamma}\\
=&\sum_{n=N'}^N \int_{Q_{N!}^{-1}(R_n)} \frac{S}{N}\frac{n}{N!} d\hat{\gamma}_n\\
=&\sum_{n=N'}^N \int_{R_n} \frac{n}{N}\hat{s} d\hat{\tau}_n\\
=&\sum_{n=N'}^N \int_{R_n} \hat{s}_n d\hat{\tau}_n\\
=&\sum_{n=N'}^N \int_{\mathcal{G}_n} s_n d\tau_n
\end{align*}

Hence, we proved \hyperref[eqn: lptotransport]{Equation \ref{eqn: lptotransport}} holds.

\subsubsection{Existence} \label{subsubsection: existenceproof}
Lastly, we prove the existence of a stable assignment. We will first prove the duality relation stated in \hyperref[theorem:finiteduality]{Theorem \ref{theorem:finiteduality}} holds and the optimizers can be achieved. Then, we establish the relationship between the duality relation and the stable assignment.

The proof of the duality relation is proceeded in three steps:
\begin{enumerate}
	\item reformulate the social welfare maximization problem as a symmetric transport problem, as stated in \hyperref[prop:linearprogrammingtotransport]{Proposition \ref{prop:linearprogrammingtotransport}}
	\item apply the duality result in multi-marginal transport problem in \cite{Kellerer84}
	\item prove the dual problem of the symmetric tranport problem is the same as a social welfare minimization problem such that no blocking coalition exists
\end{enumerate} 

To prove  \hyperref[theorem:finiteduality]{Theorem \ref{theorem:finiteduality}}, we state two lemmas. \hyperref[lemma: duality]{Lemma \ref{lemma: duality}} states that the theorem holds when there is a unique group size. \hyperref[lemma: Uhat]{Lemma \ref{lemma: Uhat}} implies the dual problem of the transport problem is the same as the social welfare minimization problem such that no blocking coalition exists. We need the following definitions to state these two lemmas:

\begin{enumerate}
	\item $\hat{\Gamma}$ is the set of measures on $I^{N!}$ such that all marginals are $\mu$:
	$$\hat{\Gamma}=\left\{\hat{\gamma}\in\mathcal{M}_+(I^{N!}):\hat{\gamma}(A\times I\times...\times I)=...=\hat{\gamma}(I\times...\times I\times A)=\mu(A), \forall A\subset I  \right\}$$
	\item 	$\hat{\mathcal{U}}$ is the set of imputations such that there is no blocking fractional group.
	$$\hat{\mathcal{U}}=\left\{u\in L^1(I,\mu): \sum_{k=1}^{N!} u(i_k)\ge (N-1)!S(i_1,...,i_{N!}), \forall i_1,...,i_{N!}\in I\right\}$$
\end{enumerate} 
For comparison, we recall that 
$$\mathcal{U}_n=\left\{u\in L^1(I,\mu): \sum_{k=1}^{n} u(i_k)\ge s([i_1,...,i_{n}]), \forall i_1,...,i_{n}\in I\right\}$$
and
$$\mathcal{U}=\left\{u\in L^1(I,\mu): \sum_{i\in G} u(i)\ge s(G), \forall G\in \mathcal{G} \right\}$$

\begin{lemma} \label{lemma: duality}
	For any symmetric upper semi-continuous function $S$ on $I^{N!}$ satisfying \hyperref[eqn: bigSbdd]{Equation \ref{eqn: bigSbdd}}, 
	$$\sup_{\hat{\gamma}\in\hat{\Gamma}_{sym}}  \int_{I^{N!}} \frac{S}{N} d\hat{\gamma}=\inf_{u\in \hat{\mathcal{U}}} \int_I u d\mu$$
	and the infimum can be achieved.
\end{lemma}
\begin{proof}
	See \hyperref[appendix:duality]{Appendix \ref{appendix:duality}}.
\end{proof}

\begin{lemma}\label{lemma: Uhat}
$$\hat{\mathcal{U}}=\mathcal{U}=\bigcap_{n=N'}^N \mathcal{U}_n$$
\end{lemma}

\begin{proof}
		See \hyperref[appendix: Uhat ]{Appendix \ref{appendix: Uhat }}.
\end{proof}

Now, we are ready to prove the duality theorem, \hyperref[theorem:finiteduality]{Theorem \ref{theorem:finiteduality}}. 
 
\begin{proof}[Proof of the duality theorem] \label{proof: duality}
By \hyperref[prop:linearprogrammingtotransport]{Proposition \ref{prop:linearprogrammingtotransport}},
$$\sup_{\tau\in\Tau} \sum_{n=N'}^N \int_{\mathcal{G}_n} s_n d\tau_n=\sup_{\hat{\gamma}\in \hat{\Gamma}_{sym}} \int_{I^{N!}} \frac{S}{N} d\hat{\gamma}$$
By \hyperref[lemma: duality]{Lemma \ref{lemma: duality}}, 
	$$\sup_{\hat{\gamma}\in\hat{\Gamma}_{sym}}  \int_{I^{N!}} \frac{S}{N} d\hat{\gamma}=\inf_{u\in \hat{\mathcal{U}}} \int_I u d\mu$$
and the infimum on the right can be achieved. By \hyperref[lemma: Uhat]{Lemma \ref{lemma: Uhat}}, 
$$\inf_{u\in \hat{\mathcal{U}}} \int_I u d\mu=\inf_{u\in \mathcal{U}} \int_I u d\mu$$
Therefore, we have the duality relation:
$$\sup_{\tau\in\Tau} \sum_{n=N'}^N \int_{\mathcal{G}_n} s_n d\tau_n=\inf_{u\in \hat{\mathcal{U}}} \int_I u d\mu$$
See \hyperref[appendix: ctsmin]{Appendix \ref{appendix: ctsmin}} for a proof of the existence of a continuous minimizer.
\end{proof}

Next, we show the maximum social welfare can be achieved:
\begin{lemma}\label{lemma: maxachieved}
	%If the surplus function satisfies \hyperref[A1]{Assumption (A1)} and  \hyperref[A2]{Assumption (A2)}, 
	There is a maximizer solving the maximization problem
	$$\sup_{\tau \in\Tau} \sum_{n=N'}^N \int_{\mathcal{G}_n} s_n d\tau_n$$
\end{lemma}
\begin{proof}
	See \hyperref[appendix: maxachieved]{Appendix \ref{appendix: maxachieved}}.
\end{proof}

In sum, we have proved the duality relation in \hyperref[theorem:finiteduality]{Theorem \ref{theorem:finiteduality}} and the fact that the optima in the duality relation \hyperref[eqn: duality]{Equation \ref{eqn: duality}} can be achieved. To show the existence of a stable assignment, it remains to the show the equivalence relation between the stable assignment and the duality relation.

\begin{lemma}\label{lemma: equivalencerelation}
	In the duality relation, \hyperref[eqn: duality]{Equation \ref{eqn: duality}}, any maximizer $\tau$ is a stable assignment and a continuous minimizer $u$ gives a corresponding imputation. 
	
	Conversely, any stable assignment $\tau$ solves the  maximization problem in the duality relation, \hyperref[eqn: duality]{Equation \ref{eqn: duality}}, and any imputation $u$ associated with this stable assignment solves the minimization problem in the duality relation, \hyperref[eqn: duality]{Equation \ref{eqn: duality}}.	 
\end{lemma}

%%%%%%%%%%%%%%%%%%%%%%%%
\begin{proof}
	See \hyperref[appendix: equivalencerelation]{Appendix \ref{appendix: equivalencerelation}}.
\end{proof}

\subsection{Examples}

\subsubsection{Non-Differentiablity of The Maximum Social Welfare Function}\label{subsubsection: nondiff}
Firstly, we illustrate the maximum social welfare function
may not be differentiable. In the finite type case, the maximum social welfare function $\Pi:\mathbb{R}^{|I|}_{+}\rightarrow \mathbb{R}$ is given by
$$\Pi(\mu)=\sup_{\tau \in\Tau} \sum_{G\in \mathcal{G}} s(G)\tau(G)$$

We study a game with two types of players. e.g. $I=\{1,2\}$. The mass of type $1, 2$ agents are $\mu_1$ and $\mu_2$ respectively. In this game, the only permitted group size is 2. That is, $N'=N=2$. In addition, all groups have zero surplus except the groups consisting both types of players, which have one unit of surplus. i.e.
$$s([i,j])=\begin{cases}
0 & \text{~if~} i= j\\
1& \text{~if~} i\neq j
\end{cases}$$
It is easy to see that the social surplus is higher if there are more groups consisting of both types of players. Therefore, we have
$$\Pi(\mu_1,\mu_2)= \min(\mu_1,\mu_2)$$
Consequently, the maximum social surplus is not differentiable at $(1,1)$ since
\begin{align*}
&\lim_{\varepsilon\rightarrow 0^+}\Pi(1+\varepsilon,1)=1=\Pi(1,1)\\
&\lim_{\varepsilon\rightarrow 0^+}\Pi(1-\varepsilon,1)=1-\varepsilon=\Pi(1,1)-\varepsilon
\end{align*}
That is, the directional derivative of $\Pi$ at $(1,1)$ is 1 along the direction $(1,0)$, but is $0$ along the direction $(-1,0)$.

However, the set of imputations coincide with the superderivatives of $\Pi$: any stable assignment assigns $\min(\mu_1,\mu_2)$ unit mass of the groups consisting players of both types, and assigns the remaining players pairwisely. Players' payoffs at $\mu=(\mu_1,\mu_2)$ is given by
$$
\begin{cases}
\text{if~} \mu_1=\mu_2,& u_1+u_2=1, u_1\ge 0,u_2\ge 0\\
\text{if~} \mu_1>\mu_2,& u_1=0, u_2=1\\
\text{if~} \mu_1<\mu_2,& u_1=1, u_2=0
\end{cases}$$

\subsubsection{Unifying Group sizes}

Next, we illustrate the trick of reformulating a problem with multiple group sizes to a problem with a unique group size. 

We study a game with three types of players. e.g. $I = \{1,2,3\}$. There are a unit mass of each type of players in the game. i.e. $\mu_1=\mu_2=\mu_3=1$. Permitted group sizes are 2 and 3. i.e. $N'=2, N=3$. Group surpluses are given by 
$$\begin{cases}
	 s([1,2,2])=9, s([3,3,3])=4\\
	s([1,2])=2,  s([3,3])=3\\
	\text{all~other~groups~have~surplus~0}
\end{cases}$$

To reformulate a problem with multiple group sizes as a problem with a unique group size, we treat a 2-person group as $2/3$ unit of some fractional group. For instance, the 2-person group $[i,j]$ corresponds to $2/3$ unit of the fractional group $[i,i,i,j,j,j]$, which consists of 1.5 units of type $i$ players and 1.5 units of type $j$ players.

Therefore, to find an optimal partition of the players into small groups of size 2 and 3, it is sufficient to find the optimal partition of the players into fractional groups. Formally, we use a 6-tuple $[i_1,...,i_6]\in I^6/\sim_6$ to denote a fractional group. Intuitively, a fractional group $[i_1,...,i_6]$ is a subset of players of total mass 3, which consists of 1/2 unit of type $i_k$ players for each $1\le k\le 6$.  In this example, there are 13 types of fractional groups.
\begin{align*}
\hat{\mathcal{G}}=\{&[1,1,1,1,1,1],[2,2,2,2,2,2],[3,3,3,3,3,3],\\
&[1,1,1,2,2,2],[1,1,1,3,3,3],[2,2,2,3,3,3],\\
&[1,1,1,1,2,2],[1,1,1,1,3,3],[2,2,2,2,1,1],\\
&[2,2,2,2,3,3],[3,3,3,3,1,1],[3,3,3,3,2,2],[1,1,2,2,3,3]
\}
\end{align*}

In particular, a fractional group might correspond to two groups. For instance, both the 2-person group $[1,1]$ and the 3-person group $[1,1,1]$ correspond to the fractional group $[1,1,1,1,1,1]$. However, when social welfare is maximized, both the 2-person group $[1,1]$ and the 3-person group $[1,1,1]$ are formed if and only if $1.5s[1,1]=s[1,1,1]$. When $1.5s[1,1]>s[1,1,1]$, only the 2-person group $[1,1]$ will be formed as we can break the group  $[1,1,1]$ into 1.5 units of the 2-person group $[1,1]$ and the total surplus will increase. When $1.5s[1,1]<s[1,1,1]$, similar operations can be implemented to increase the social welfare. Therefore, we define the surplus of a fractional group to be the maximum total surplus this fractional group of players can obtain by forming 2- or 3-person groups. Numerically, the surplus function is given by:
$$
\hat{s}(\hat{G})=
\begin{cases}
1.5\times 2=3, & \text{~if~} \hat{G}=[1,1,1,2,2,2]\\
9, & \text{~if~} \hat{G}=[1,1,2,2,2,2]\\
\max(4,1.5\times 3)=4.5, & \text{~if~} \hat{G}=[3,3,3,3,3,3]\\
0, &\text{~otherwise~} \\
\end{cases}$$ 

It is worth noting that we define the surplus of the fractional group $\hat{G}=[1,1,1,2,2,2]$ to be the welfare of the group, consisting of 1.5 units of type 1 players and 1.5 units of type 2 players, when all players are assigned into 2-person groups, as no 3-person group corresponds to the fractional group $\hat{G}=[1,1,1,2,2,2]$. However, the maximum welfare of this group of players is at least $2/3\times 9=6$ since it contains 2/3 mass of 3-person group $[2,3,3]$. Consequently, $\hat{G}$ will not be formed (appear with positive mass in the partition) in any optimal partition of players into fractional groups. Besides, the average surplus of the 2-person group $[3,3]$, which is 3/2, is larger than the average surplus of the 3-person group $[3,3,3]$, which is 1. Therefore, all type 3 players will be assigned in pairs in any stable assignment.

Given the surplus of fractional groups, the welfare maximization problem is reformulated as 
$$\sup_{\hat{\tau}} \sum_{\hat{G}\in \hat{\mathcal{G}}}\hat{s}(\hat{G})\hat{\tau}(\hat{G}) $$

Here, $\hat{\tau}: \hat{\mathcal{G}}\rightarrow \mathbb{R}_+$ is the non-negative valued function on the set of fractional groups satisfying the following consistency condition: for any disjoint $i,j,k\in \{1,2,3\}$,
\begin{align*}
3\tau([i,i,i,i,i,i])+2 [\tau([i,i,i,i,j,j])+ \tau([i,i,i,i,k,k])]&\\
+1.5[\tau([i,i,i,j,j,j])+\tau([i,i,i,k,k,k])]&\\
+[\tau([i,i,j,j,j,j])+\tau([i,i,k,k,k,k])+\tau([i,i,k,k,j,j])]&=1
\end{align*}

That is, we reformulate the problem in which groups sizes are 2 or 3 to a problem in which the unique group size is $3!=6$.

\subsubsection{Unbounded Group Sizes}
When group sizes are unbounded, the existence of a stable assignment has been studied in an approximate manner\footnote{with the approximate feasibility condition.} in  \cite{hammond1989continuum}. In this subsection, we illustrate that an stable assignments may not exist when group sizes are unbounded. 

Consider a game with a continuum of homogeneous players. The surplus of any $n$-person group is given by 
$s_n=n-1$. By the equal treatment property of stable assignment, the payoff of any player, in any stable assignment, cannot be higher than or equal to 1. If there is a stable assignment in which all agents have a payoff $u<1-1/m$. Then, $m$ players will form a blocking coalition.  

In this example, the surplus function is not uniformly bounded: we have $s_n\rightarrow +\infty$ as $n\rightarrow +\infty$. When the surplus function is uniformly bounded, the uniform bound will impose a natural upper bound on the group sizes: for formed groups containing too many players, the average payoff of players will goes to zero. Therefore, there is a subset of players in this formed group who have incentive to form a blocking coalition.

\subsubsection{An Exchange Economy with Groupwise Externalities}\label{subsubsection: envy}

In exchange economies, when there is no consumption externalities, the surplus function is always super-additive. Here, the surplus of a group is defined by the maximum total welfare of the group members who exchange commodities within the group. In contrast, when there are consumption externalities, the surplus function may no longer be super-additive. We give one such example with groupwise envy. This type of groupwise externality differs from the widespread externalities introduced in \cite{hammond1989continuum}.

Now, we describe an economy with groupwise envy. In the economy, agents, with quasi-linear utility functions, either do not trade or trade in pairs. There are one consumption good and two types of agents, type 1 and type 2. Utility functions and initial endowments are given by 
$$u_1(m,x,G)=m_1+\sqrt{x_1}+(x_1-\max_{j\in G} x_j), \omega_1=(0,0)$$
$$u_2(m,x,G)=m_2+100\sqrt{x_2}+(x_2-\max_{j\in G} x_j), \omega_2=(0,100)$$

Here, utility functions are real-valued maps on the monetary good consumption, consumption good consumption and the group the agent is in. The externality term $x_i-\max_{j\in G} x_j$ suggests that an agent would envy his trade partner if his trade partner consumes more consumption goods than him.

By the definition of the surplus function,
$$s([1])=0, s([2])=1000$$
$$s([1,1])=0, s([2,2])=2000$$
$$s([1,2])=\max_{x_1,x_2\ge 0: x_1+x_2=100} \sqrt{x_1}+100\sqrt{x_2}+100-2\max(x_1,x_2)\simeq 900.083$$
where the maximizers are  $x_1\simeq 0.03, x_2\simeq 99.97$.

That is, the joining of a type 1 agent to the 1-person group consisting of just a type 2 agent will reduce the group surplus due to the existence of groupwise envy. Intuitively, since type 2 agents value the consumption good much more than type 1 agents, in any efficient outcome, almost all consumption good will be consumed by type 2 agents. Consequently, although the type 1 agent consumes a little consumption good by trading with the type 2 agent, his envy makes him much less happy. In addition, the type 2 agent, consuming a little less consumption good, is also less happy. Therefore, with groupwise envy, the sum of payoffs are decreased after the trade.

%\subsection{An algorithm for the finite type case}

\subsection{Applications}
In this part, we apply our results to three problems.

\subsubsection{Surplus Sharing Problem}

In a market, a continuum of agents of different types form groups in order to share group surpluses. Each agent can participate to form only one group and each group can be formed by at most $N$ agents. The surplus of a formed group only depends on the type distribution of its members.  

By our results, there is a stable state in this market. At the stable state, agents are endogenously segmented into small groups and they have a way to share the group surplus such that no subset of agents have incentives to form a new group. 

We have three more observations regarding the stable state of the market. Firstly, at the stable state, the same type of agents have the same payoff, which can be understood as the wage level of this type of players. Secondly, at the stable state, the surplus of any formed group is equal to the total payoff of its members. Thirdly, at a stable state, even if an agent can initiate to form a new group and obtain all the group surplus of the new group by paying some other agents an amount no less than their wage levels, no agent has incentive to do so. The reason of this fact is that, at an stable state, the group surplus an agent could obtain by forming a new group minus all the wages he pays to the others cannot be larger than his current wage level.

\subsubsection{Production Problem} \label{subsubsection: production}

Given a bag of (finitely many or infinitely many) inputs (Lego bricks), there are numerous ways to combine them into outputs (toys), each of which has a value. The production of any output requires a specific finite combination of inputs (Lego bricks). Production technology specifies a relation between input distributions and output values. Knowing the production technology, there are two natural questions. Firstly, what is the bag's value, defined as the maximum total value of outputs obtained by using the inputs in the bag? Secondly, how to achieve this maximum value?

\begin{figure}[h]
\centering	
\includegraphics[width=0.5\textwidth]{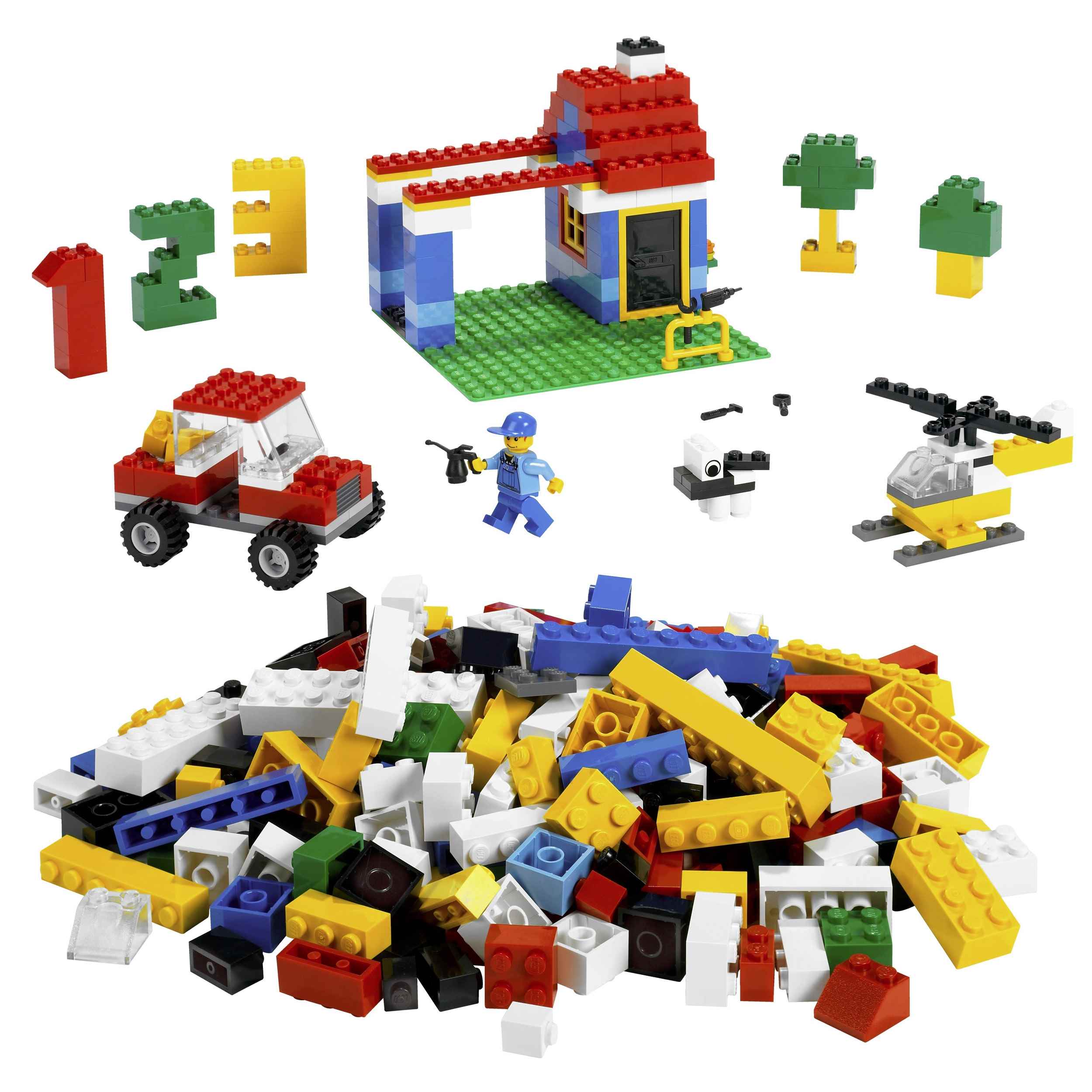}
\caption{Production problem\protect\footnotemark}
\label{fig:production problem}
\end{figure}

\footnotetext{an online picture.}

By our results, each type of inputs is endowed with a value defined by the imputation corresponding to a stable assignment. Conceptually, knowing these values, we can answer the first question on the bag's value in the following sense: when the bag contains a continuum of inputs, the sum (integration) of input values in the bag will give the bag's value. When the bag contains only finitely many inputs, the sum of input values in the bag will give an upper bound on the bag's value. When we have a large number of bags containing the same inputs, the total value of these bags is approximately equal to the total value of all inputs.

Actually, it is easy to see that these two questions can be answered by directly solving a profit maximization problem stated in \hyperref[eqn: finitemax]{Equation \ref{eqn: finitemax}}, in which the choice set is the set of all partitions of inputs into small groups. Each small group in the partition will be used to produce an output. When there are a continuum of inputs in the bag, the maximization problem in \hyperref[eqn: finitemax]{Equation \ref{eqn: finitemax}} gives exact answers to both questions. When there are only finitely many inputs, the maximization problem in \hyperref[eqn: finitemax]{Equation \ref{eqn: finitemax}} gives approximate answers to both questions.

However, it is usually computationally infeasible to solve the linear programming problem \hyperref[eqn: finitemax]{Equation \ref{eqn: finitemax}}, as the number of unknowns $|\mathcal{G}|$ becomes astronomical even when the maximum group size is very small.\footnote{$|\mathcal{G}|=\sum_{n=N'}^N \multiset{|I|}{n}=\Theta(|I|^N)$.} For example, when there are $|I|=1000$ types of inputs and each toy requires a combination of at most $4$ pieces of bricks, there are
$$|\mathcal{G}|=\sum_{n=1}^4 \multiset{1000}{n}\simeq 4.2\times 10^{10}$$
types of toys. That is, in order to answer either question, the direct method is to solve a maximization problem with $4.2\times 10^{10}$ unknowns.
 
To  deal with the curse of dimensionality, a first thought is that we could apply the duality relation in \hyperref[theorem:finiteduality]{Theorem \ref{theorem:finiteduality}} to answer the first question: to compute the bag's value, or a upper approximation of the bag's value, we only need to compute the values of inputs. That is, we only need to solve the dual problem in \hyperref[eqn: finitemin]{Equation \ref{eqn: finitemin}} with only $|I|$ unknowns.
However, this dual problem is also hard to solve as there are $|\mathcal{G}|$ constraints in it. Moreover, the minimization problem provides no answer to our second question about the optimal way of production.

It appears a better way to deal with the curse of dimensionality is to use \hyperref[prop:linearprogrammingtotransport]{Proposition \ref{prop:linearprogrammingtotransport}} to reformulate the maximization problem as an $N!$-marginal symmetric transport problem. Utilizing the symmetric structure in the symmetric transport problem, \cite{friesecke2018breaking} proved this $N!$-marginal symmetric transport problem can be further reformulated as a combinatorial maximization problem with only $(N!+1)|I|$ unknowns.\footnote{The number of unknown can be further reduced if not all group sizes are permitted. In particular, we can always replace the number $N!$ by the least common multiple of all group sizes. When $N$ is fixed and $|I|$ is large, $(N!+1)|I|=\Theta(|I|)$. Moreover, we note this reduction technique is not a free lunch. In particular, it transforms a linear programming problem to a combinatorial optimization problem. However, when $|I|$ is large, solving the linear programming problem is impossible due to the large number of unknowns. This reduction trick provides a possibility to solve the optimization problem even though the transformed problem is combinatorial. For instance, \cite{khoo2019convex} used the trick to solve a problem which has $10^{25}$ unknowns initially.} In the example with 1000 types of inputs and group sizes are up to 4, the reformulated maximization problem has 25000 unknowns. We refer to \cite{khoo2019convex} for a numerical implementation of this dimension reduction trick.

\subsubsection{Market Segmentation}
Lastly, we consider an exchange economy with no centralized market. In such economy, a continuum of agents exchange commodities with each other within small groups of bounded finite sizes. Each type of agent is represented by a continuous quasi-linear utility function and an initial endowment. In this economy, we allow the existence of groupwise externalities: agents' utility functions may depend on the consumption of their trade partners. An example of an economy with groupwise externalities is given in \hyperref[subsubsection: envy]{Section \ref{subsubsection: envy}}. 

Our result suggests, with necessary assumptions on the continuity and boundedness of utility functions and initial endowments, there is a stable state in which the market is segmented into small groups and no subset of agents can jointly do better by forming a new group to trade with each other.

%In \hyperref[section: economy]{Section \ref{section: economy}}, we will discuss segmented exchange markets with no externalities in details.

\subsection{Related Literature}

We review the literature from two perspectives: the conceptual perspective and the methodological perspective.

From the conceptual perspective, the most relevant literature to the material in this section is the f-core literature in \cite{kaneko1982cores}, \cite{kaneko1986core}, \cite{hammond1989continuum}, \cite{KanekoWooders96}. With an approximate feasibility condition, this sequence of works proves the existence of an approximately stable assignment in games with a compact type space.\footnote{This sequence of works are mainly focused on models with non-transferable utility. Here, I am talking about the application of these works to a game with transferable utility. It is worth noting that the method we developed in this paper depending on linear dualities, and thus cannot be applied to nontransferable utility games directly. For a link between duality and games with nontransferable utility, we refer to \cite{noldeke2018implementation}.} 

There are three major differences between the literature mentioned above and the work in this section. Firstly, in this section, we developed a new formulation for the concept f-core by studying a statistical representation of the partition. Secondly, dropping the approximate feasibility assumption, we proved the existence of an  stable assignment in a game with possibly a non-compact type space. Thirdly, we showed that, when the surplus function is continuous, similar agents obtain similar payoffs in a stable state.  

The framework and results in this section unified the literature on matching problem, roommate problem and f-core. When the unique permitted group size is 2 and the surplus function has a bipartite structure, our work implies results on the matching problem. In particular,  \cite{ShapleyShubik71} and \cite{roth1993stable} studied matching problem with a finite type space. \cite{gretsky1992nonatomic}, \cite{chiappori2010hedonic} and \cite{chiappori2016multidimensional} studied matching problem with a general type space. When the permitted group size is $k$ and the surplus function is assumed to have a k-partite structure, our work implies results on the multi-matching problem such as \cite{carlier2010matching} and \cite{pass2015multi}. When there is an upper bound on group sizes, our work implies results on f-core in the transferable cases. In particular, \cite{kaneko1982cores}, \cite{kaneko1986core},  \cite{wooders2012theory} studied f-core with a finite type space. \cite{KanekoWooders96} studied f-core with a compact atomless type space. 

Our model assumes that there are a continuum of players in the game. It is valid to ask, with a {\it discrete} player set, whether a stable assignment exists or not. To my knowledge, our understanding is complete only when the unique permitted group size is 2: when the surplus function has a bipartite structure, \cite{koopmans1957assignment} proves the existence of a stable assignment. When the surplus function has no bipartite structure, \cite{chiappori2014roommate} proves the existence of a stable assignment if the player set is replicated 2 times. Both observations rely on the validity of Birkhoff-von Neumann theorem on matrices. Finding a high dimensional analog of Birkhoff von-Neumann theorem is an ongoing project\footnote{See \cite{linial2014vertices} for related literature.}. Therefore, we need to develop new tools to understand discrete models when the maximum group size is larger than 2.

From the methodological perspective, the literature on f-core such as \cite{kaneko1982cores}, \cite{kaneko1986core}, \cite{KanekoWooders96} did not use the optimization method to prove the existence of a stable assignment, as the problem is formulated in a non-transferable utility environment. However, in a transferable utility environment, it is most natural to use linear programming to analyze the problem. Indeed, when there are only finitely many types of players in the game, it is well known that linear programming helps to prove the existence of a stable assignment. We refer to \cite{wooders2012theory} for details. The idea of using linear programming to prove stability dates back to the ground-breaking works of  \cite{bondareva1963some} and \cite{shapley1967balanced}, in which the connection between
core and linear programming is established. Several works including \cite{kannai1969countably}, \cite{schmeidler1972cores}, \cite{kannai1992core} extended this connection to infinite type spaces. However,  these extensions are not applicable to our model in which group sizes are bounded.

In this section, we connected the literature on f-core to transport problems. This connection helps us to prove the existence of a stable assignment for general type space and significantly reduces the number of unknowns for finding stable assignments  in the finite type case. 

Actually, there is a long tradition of applying the transportation method to study stability. In this line, most works assume some structure on the surplus function:  when the surplus function has a bipartite structure and the unique permitted group size is 2, \cite{koopmans1957assignment}, \cite{ShapleyShubik71} studied stable matching with a finite type space. \cite{gretsky1992nonatomic}, \cite{chiappori2010hedonic}, \cite{chiappori2016multidimensional} studied stable matching with a general type space. When the surplus function has a $k$-partite structure and the unique permitted group size is $k$, \cite{carlier2010matching} and \cite{pass2015multi} studied stable multi-matching with general type spaces. When the surplus function has no special structure and the unique permitted group size is 2, \cite{chiappori2014roommate} studied stable matching in a discrete player model. Their key observation is that any pure matching in the replicated player space corresponds to an assignment. In this section, we propose an identification trick that is inversely related to their observation: an assignment can be identified by some symmetric transport plan in a replicated player space. This identification trick helps us in two ways. Firstly, it helps us to study games with a continuum of players, with small groups containing more than two players and with small groups of different sizes. Secondly, it helps to establish the equivalence relation between the set of stable assignments and the solution of a symmetric optimal transport problem.

Recently, \cite{friesecke2018breaking} introduced a dimension reduction trick to significantly reduce the number of unknowns in a symmetric optimal transport problem. This trick has a huge potential in computations. In particular, \cite{khoo2019convex} used this dimension reduction trick to solve a problem with $10^{25}$ unknowns. In this section, by relating stable assignments to symmetric transport plans, we enable the possibility of applying this dimension reduction trick to find stable assignments. 

For more complete surveys on optimal transport problems, we refer to \cite{Villani08} for 2-marginal problems, to \cite{pass2011structural} and \cite{Pass15} for multi-marginal problems.

\section{Games with Positive-Size Groups}
In \hyperlink{section.2}{Section 2}, a small group is defined as a finite subset of players. Therefore, every small group contains at most finitely many types of players and is of ``measure zero". In some applications such as \cite{Schmeidler72}, small groups may contain infinitely many types of players and have small masses. To study this phenomenon, in this section, we extend our previous model and study small groups with small positive masses. %We use two positive numbers $0<\varepsilon'\le\varepsilon\le 1$ denote the bounds of the group sizes.

%However, it is easy to see a partition of the continuum,  each component of which in the partition has at least a mass of $\varepsilon'$, may not exist. For example, there is no way to partition the interval $[0,1]$ into subsets with Lebesgue measure $2/3$. Thus, similar to \cite{scarf1967core}, we define assignments as a fractional notion: there are 2/3 units of the group $[0,2/3]$, 2/3 units of the group $[1/3,1]$ and 2/3 units of the group $[0,1/3]\cup[2/3,1]$. This fractional notion can be interpreted in a probabilistic language or as a partition of the infinite replication of the continuum of players. 

This section is organized as follows. In \hyperlink{subsection.3.1}{Section 3.1}, we set up the problem. In \hyperlink{subsection.3.2}{Section 3.2}, we state the results. In \hyperlink{subsection.3.3}{Section 3.3}, we review the literature. All proofs are postponed to \hyperref[appendix: sect3]{Appendix \ref{appendix: sect3}}.

\subsection{Model}
We study a cooperative (transferable utility) game represented by the tuple $((I,\mu),s, \varepsilon',\varepsilon)$. Here, the positive numbers $\varepsilon', \varepsilon \in (0,1]$ are the lower and upper bounds on group sizes respectively.

\subsubsection{Type Space}
The {\it type space} of players is summarized by a probability space $(I,\mu)$. In the tuple, $I$ is a  Polish space representing the set of players' types and $\mu\in \mathcal{P}(I)$ is a probability measure on $I$ representing the distribution of players' types. In particular, $\mu$ is not assumed to be atomless.

\subsubsection{Groups}
Groups are the units in which players interact with each other. Similar to the formulation of sub-population in \cite{che2019stable}, a group is defined by a non-negative measure $\nu\in \mathcal{M}_+(I)$ on $I$ that is not larger than $\mu$. The size of a group $\nu$ is its total measure $\norm{\nu}=\nu(I)$ on $I$. All group sizes are bounded from below by $\varepsilon'\in (0,1]$\footnote{Technically, it is important to assume that the group sizes are bounded away from zero. The strictly positive lower bound will imply that the set of assignments defined later is a compact set in weak topology.} and bounded from above by $\varepsilon\in (0,1]$, where $\varepsilon\ge\varepsilon'$. When $\varepsilon=\varepsilon'$, only one group size is permitted in the game.

Formally, a {\it group} is a non-negative measure $\nu$ on $I$ such that $\nu\le \mu$ and $\norm{\nu}\in [\varepsilon',\varepsilon]$. {\it The set of groups} $\mathcal{G}$ consists of all groups which sizes are bounded by $\varepsilon'$ and $\varepsilon$:
$$\mathcal{G}=\{\nu\in \mathcal{M}_+(I):\nu\le \mu, \varepsilon'\le \nu(I)\le \varepsilon\}$$

The set of groups $\mathcal{G}$ is endowed with the weak topology\footnote{By Theorem 8.3.2 in \cite{Bogachev2}, the weak topology on $\mathcal{M}_+(I)$ is metrizable by the Kantorovitch-Rubinstein norm $\norm{\cdot}_0$ defined by
$$\norm{\mu}_0=\sup\left\{\int_I f d\mu: f\in Lip_1(I), \norm{f}_{\infty}\le 1  \right\}$$}.

\subsubsection{Surplus Function}
The surplus function specifies the total amount of surplus group members can share. Formally, a {\it surplus function} $s:\mathcal{G}\rightarrow \mathbb{R}_+$ is a real valued function on the set of groups. There are two assumptions on the surplus function:

\begin{enumerate}
	\item[(B1)]\label{B1} $s$ is upper semi-continuous in the weak topology.
	\item[(B2)]\label{B2} there is a bounded continuous function $a\in C_b(I)$ such that $s(\nu)\le \int_I a d\nu$, for all $\nu\in \mathcal{G}$.
\end{enumerate}

\hyperref[B1]{Assumption (B1)} states the continuity requirement of the surplus function and \hyperref[B2]{Assumption (B2)} ensures the integrability of the surplus function. We note \hyperref[B2]{Assumption (B2)} is satisfied if there is a constant $c\ge 0$ such that $s(\nu)\le c\norm{\nu}$ for all $\nu\in \mathcal{G}$. That is, \hyperref[B2]{Assumption (B2)} is satisfied if the average contribution of a player to any group is bounded from above by a constant $c>0$.

We say a surplus function is {\it linear}, if for all $\nu\in\mathcal{G}$ and $\alpha>0$, we have $s(\alpha \nu)=\alpha s(\nu)$ whenever $\alpha\nu\in \mathcal{G}$. Moreover, we say a surplus function is {\it super-additive}, if for all $\nu_1,\nu_2\in\mathcal{G}$, we have $s(\nu_1+\nu_2)\ge s(\nu_1)+s(\nu_2)$ whenever $\nu_1+\nu_2\in\mathcal{G}$. In particular, a surplus function generated from a general equilibrium model with no externalities is both linear and super-additive.

\subsubsection{Assignments}

An assignment is a partition of a continuum replication of players into groups of permitted sizes. The continuum replication of players is consistent with our study of small groups: a group of players containing 99\% of the total population is small in the continuum replicated player set.

Formally, an {\it assignment} is a non-negative measure on the set of groups $\mathcal{G}$ satisfying the consistency condition:
$$\int_\mathcal{G} \nu(A) d\gamma(\nu)=\mu(A), \forall \text{measurable~} A\subset I\footnote{For any Borel measurable  $A\subset I$, the real valued map $\nu\rightarrow \nu(A)$ on $\mathcal{M}_+(I)$ is continuous in the weak topology since 
	$$|\nu_1(A)-\nu_2(A)|\le \int_A 1 d|\nu_1-\nu_2|\le \int_I 1 d|\nu_1-\nu_2|\le \norm{\nu_1-\nu_2}_0$$
	Therefore, the map $\nu\rightarrow \nu(A)$ is Borel measurable. That is, $\int_\mathcal{G} \nu(A) d\gamma(\nu)$ is well-defined.} $$

The consistency condition states that all players, whose types are in $A$, are assigned by the assignment $\gamma$.  Moreover, we use a set $\Gamma_{\mathcal{G}}$ to denote {\it the set of assignments}. i.e.
$$\Gamma_{\mathcal{G}}=\left\{\nu\in \mathcal{M}_+(\mathcal{G}): \int_\mathcal{G} \nu(A) d\gamma(\nu)=\mu(A), \forall \text{measurable~} A\subset I\right\}$$
%Similar to the previous section, to get rid of the limit points of formed groups\footnote{Intuitively, a group is formed if it has a positive mass in a partition represented by $\gamma$.} in the set $\supp(\gamma)$, we say a measurable set $\mathcal{F}\in \supp(\gamma)$ is a {\it set of formed groups under assignment $\gamma$} if $\gamma(\mathcal{F})=\gamma(\mathcal{G})$.

Intuitively, $\gamma(\nu)$ gives the average mass of group $\nu$ in the partition of the replicated player set: when the player set is replicated by a large number, the product of $\gamma(\nu)$ and the number of replication
specifies the number of groups $\nu$ in the partition of the replicated player set determined by assignment $\gamma$. 

We finish this subsubsection by two remarks. 

Firstly, the set of assignments $\Gamma_{\mathcal{G}}$ is non-empty since  $\frac{2}{\varepsilon+\varepsilon'}\mathbb{1}_{\frac{\varepsilon+\varepsilon'}{2} \mu} \in \Gamma_{\mathcal{G}}$. That is, all players are assigned to the group $\frac{\varepsilon+\varepsilon'}{2}\mu\in \mathcal{G}$. 

Secondly, an assignment is not a probability measure on $\mathcal{G}$ in general. We prove by contradiction. If an assignment is a probability measure on $\mathcal{G}$, then, by taking $A=I$ in the consistency condition, we have that the left hand side of the consistency equation is less than $\varepsilon$ while the right hand side of the consistency equation is 1. Contradiction.

\subsubsection{Stability}
Similar to the case where group sizes are finite, we say an assignment $\gamma\in \Gamma_{\mathcal{G}}$ is {\it stable} if there is an imputation $u\in L^1(I,\mu)$ such that
\begin{enumerate}
	\item $\int_I u d\nu \le s(\nu)$, for $\gamma-$almost all groups $\nu\in \mathcal{G}$\footnote{Here, we need the term ``almost everywhere" since we did not prove that there exists a continuous imputation. See \hyperref[footnote: redefformedgroups]{Footnote \ref{footnote: redefformedgroups}} for the related discussions on ``formed groups".}
	\item $\int_I u d\nu \ge s(\nu)$, for all groups $\nu\in\mathcal{G}$
\end{enumerate}

The first condition is the feasibility condition, which states that, in any formed group (intuitively, a group with a positive mass in the partition), group members can in total share no more than the group surplus. The second condition is the no-blocking condition, which states no group of agents could jointly do better by forming a new group.

We note, when the surplus function is linear, a stable assignment can be interpreted as a stable way to partition of the continuum players into small groups with no lower bound on group sizes\footnote{When the surplus function is linear, the surplus of a group $\nu$ is equal to the total surplus of $1/\alpha$ units of groups $\alpha\nu$. Consequently, in any stable assignment $\gamma$, for any formed group $\nu$ and $\alpha>0$, $\int_{I} u d(\alpha \nu)\le s(\alpha \nu)$ provided the scaled group $\alpha \nu\in \mathcal{G}$.  Therefore,  in a stable assignment, any group with a positive mass in the partition can be split in an arbitrary uniform way without changing the stability property of the assignment.}.

\subsection{Results}
We prove the existence of a stable assignment by establishing a ``continuum"-marginal extension of the Koopmans-Kantorovitch duality theorem. Therefore, we have the following two theorems.

\begin{theorem}\label{theorem: infdimexistence}
	For any game $((I,\mu),s, \varepsilon', \varepsilon)$ satisfying \hyperref[B1]{Assumption (B1)} and  \hyperref[B2]{Assumption (B2)}, there is a stable assignment.
\end{theorem}

\begin{theorem}\label{theorem: infdimduality}
	For any game $((I,\mu),s, \varepsilon', \varepsilon)$ satisfying \hyperref[B1]{Assumption (B1)} and  \hyperref[B2]{Assumption (B2)}, 
	$$\sup_{\gamma\in\Gamma_{\mathcal{G}}} \int_\mathcal{G} s d\gamma=\inf_{u\in\mathcal{U}}\int_I u d\mu$$
	where $\mathcal{U}=\{u\in L^1(\mu): \int_I u d\nu\ge s(\nu), \forall \nu\in\mathcal{G} \}$ and the infimum could be attained.
\end{theorem}

The proofs of these two theorems are postponed to \hyperref[appendix: sect3]{Appendix \ref{appendix: sect3}}.

\subsection{Related Literature}

Small groups of positive sizes have been studied by \cite{Schmeidler72} in exchange economies. In his work, Schmeidler proved that any core allocation cannot be blocked by any small group of epsilon sizes. However, since core allocations are not feasible in general, this result is not enough to imply the existence of a stable assignment.

In this section, we developed a game model with small groups of positive sizes and proved the existence of a stable assignment. In particular, we do not assume the surplus function to be linear or super-additive. Therefore, our model can be used to analyze exchange economies with externalities.

To prove the existence result, we generalized the
Koopmans-Kantorovitch duality theorem to a ``continuum"-marginal case. Our proof is built on the proof of duality theorem in multi-marginal transport problem in \cite{Kellerer84} and \cite{dudley2002real}.

\section{Conclusions}
In this paper, we study a game with a continuum of agents who form small groups in order to share group surpluses. Group sizes are exogenously bounded by natural numbers or percentiles. We prove that there exists a stable assignment, where no group of agents can jointly do better. Conceptually, our work provides the only existence result to this problem on our level of generality, as well as a uniform way to understand diverse solution concepts, such as stable matching, fractional core, f-core, and epsilon-sized core. Computationally, when there are finitely many types of players and group sizes are finite, we reduce the number of unknowns in the problem of finding stable assignments from about $|I|^N$ to about $|I|$, where $|I|$ is the number of player types, $N$ is the maximum  group size and $|I|$ is much larger than $N$. We achieve this reduction by reformulating the welfare maximization problem as a symmetric transport problem.

\addcontentsline{toc}{section}{References}

\bibliographystyle{abbrvnat}
\bibliography{RS}

\appendix
\section{Notations}
In this section, we clarify notations we used in this paper. In the following definitions, $I$ is always assumed to be a Polish (complete separable metric) space.

\begin{itemize}
%	\item $\mathcal{B}(I)$: Borel sigma-algebra of I  
	\item $\Theta(x^k)$: the growth rate of a function on $\mathbb{N}$ is exactly $x^k$ (See Remark 1)
	\item $\mathcal{M}_+(I)$: the space of non-negative Borel measures on I
	\item $\mathcal{P}(I)$: the space of probability measures on $I$
	\item $S_n$: the set of bijective maps from $\{1,2,...,n\}$ to itself
	\item $USC(I)$: the set of upper semi-continuous continuous functions on $I$
	\item $LSC(I)$: the set of lower semi-continuous functions on $I$
	\item $C(I)$: the set of continuous functions on $I$
	\item $C_b(I)$: the set of bounded continuous functions on $I$
	\item $Lip_1(I)$: the set of 1-Lipchitz continuous functions on $I$
	\item $L^1(I,\mu)$: the set of integrable functions on $(I,\mu)$
	\item $\norm{f}_\infty$: the essential suppremum of the function $f$
	\item $\norm{\nu}$: the total measure of a nonnegative measure $\nu$ on I
	\item $\mu_1\le \mu_2$: $\mu_1$ is less than equal to $\mu_2$ on all measurable sets
	\item $\supp(\mu)$: the support of measure $\mu$, see (Remark 2)
	\item $T_\#\mu$: the push-forward measure of $\mu$ under the map $T$ (See Remark 3)
	\item Symmetric measure: a measure $\gamma$ on $I^n$ is symmetric if 
	$$\gamma(A_1,...,A_n)=\gamma(A_{\sigma(1)},...,A_{\sigma(n)}), \forall A_1,...,A_n\subset I, \sigma\in S_n$$
\end{itemize}

\noindent {\bf Remark 1}: For a function $f:\mathbb{N}\rightarrow \mathbb{R}$, $f(x)=\Theta(x^k)$ if $cx^k\le f(x)\le C x^k$ for some constants $c, C>0$. Since big $\Theta$ notation characterizes functions according to their growth rates, different functions with the same growth rate may be represented using the same big $\Theta$ notation.

\noindent {\bf Remark 2}: For a measure $\mu$ on $I$, $\supp(\mu)$ is a closed set satisfying: 1. $\mu((\supp \mu)^c)=0$. 2. for every open $G$ intersecting $\supp(\mu)$, $\mu(G\cap \supp(\mu))>0$. When $I$ is a Polish space, $\mu$ has a unique support.\footnote{Definitions and results are taken from \cite{aliprantis2007infinite}.}

\noindent {\bf Remark 3}: For Polish spaces $I, J$, a measure $\mu$ on $I$ and a measurable function $T$ from $I$ to $J$, $T_\#\mu$ is a measure on $J$ such that for any measurable $A\subset J$, $T_\#\mu(A)=\mu(T^{-1}(A))$.\footnote{This definition is taken from \cite{Villani08}.}

\section{Omitted Proofs in Section 2}

\subsection{Metrizability of The Quotient space $I^n/\sim_n$}\label{appendix:equivalenceclass}
In this subsection, we show the quotient space $I^n/\sim_n$, defined in \hyperref[subsection: groups1]{Section \ref{subsection: groups1}}, is metrizable. In this proof, we omit the subscript of the equivalence relation and write it as $\sim$. 

To start with, it is well known that the quotient space of a metric space $(I,d_0)$ is endowed with a pseudo-metric defined by
$$d([x],[y])= \inf \sum_{k=1}^n d_0(p_k,q_k)$$
where the choice set is given by the set of all finite sequence $(p_k,q_k)_{k=1}^n$ such that $p_k,q_k\in I$, $p_1=x$, $q_n=y$ and $q_k\sim p_{k+1}$ for all $k\in \mathbb{N}$. Now, we prove the pseudo-metric $d$ defined above is a metric 
on $I^n/\sim_n$. 

\begin{lemma}\label{lemma: quotientmetric}
$$d([x],[y])=\min_{y'\sim y} d_0(x,y')$$
\end{lemma}

This lemma will imply that  $d$ is a metric on the quotient space: for any $y\in I^n$, there are only finite many $y'\in I^n$ such that $y\sim y'$. Therefore, $d([x],[y])=0$ implies $d_0(x,y')=0$ for some $y'\sim y$. Hence, $x\sim y$.

\begin{proof}
Firstly, by definition, we have $d([x],[y]) = d([x],[y'])\le d_0(x,y') $ for all $y'\sim y$. Therefore,
$d([x],[y]) \le \min_{y'\sim y} d_0(x,y') $. Now we prove the converse $d([x],[y]) \ge \min_{y'\sim y} d_0(x,y') $. To prove the converse, we use the distance preserving property of the permutation: the distance of two points  in $I^n$  are not changed if their indices are permuted by the same permutation. Consequently, for any sequence $(p_k,q_k)_{k=1}^n$ in the choice set, there is a sequence of $(r_k)_{k=2}^n$ in $I^n$, defined inductively, such that  
$$r_1=q_1$$
$$d_0(p_k,q_k)=d_0(r_{k-1},r_k), \forall k\ge 2$$
$$r_k\sim q_k$$
Therefore,
$$\sum_{k=1}^n d_0(p_k,q_k)= d_0(x,r_1)+\sum_{k=2}^n d_0(r_{k-1},r_k)\ge d(x,r_k)$$
Since $r_n\sim q_n=y$,
$$\sum_{k=1}^n d_0(p_k,q_k)\ge \min_{y'\sim y} d_0(x,y')$$
As the choice of $(p_k,q_k)_{k=1}^n$ in the choice set is arbitrary, the lemma is proved.
\end{proof}

\subsection{Measurability of $\mathcal{G}_{n}(A,k)$}
\label{appendix: Gnkmeasurable}
In this subsection, we show the $\mathcal{G}_{n}(A,k)$, defined in \hyperref[subsection: assignments1]{Section \ref{subsection: assignments1}}, is measurable in $I^n/\sim_n$. 

Firstly, we prove for any Borel measurable $A_1,...,A_n\subset I$, 
$$Z=\{[i_1,...,i_n]: i_k\in A_k, \forall 1\le k\in n\}$$ is Borel measurable in $I^n/\sim_n$. By definition of Borel measurability, we only need to prove the cases that $A_1,...,A_n$ are open: fix $i_1,...,i_n\in A_1\times A_2\times...\times A_n$, for any $\delta>0$ and $[j_1,...,j_n]\in I^n/\sim_n$, if $d([i_1,...,i_n], [j_1,...,j_n])<\delta$, by \hyperref[lemma: quotientmetric]{Lemma \ref{lemma: quotientmetric}}, there is a permutation $\sigma\in S_n$ such that 
$$d_0((i_1,...,i_n), (j_{\sigma(1)},...,j_{\sigma(n)}))<\delta$$
Therefore, $d_I(i_m, j_{\sigma(m)})<\delta$ for all $1\le m\le n$. Here, $d_I$ is the metric in $I$. On the other hand, all $A_m$ are open. By taking $\delta$ small enough, we have $j_{\sigma(m)}\in A_m$. Therefore, $[j_1,...,j_n]\in Z$. Consequently, $Z$ is open, thus Borel measurable.

Next, for any $1\le k\le n$ and measurable $A\subset I$, we take $A_1=...=A_k=A$ and $A_{k+1}=...=A_n=A^c$, we have
$$\mathcal{G}_{n}(A,k)=\{[i_1,...,i_n]: i_k\in A_k, \forall 1\le k\in n\}$$
Therefore, $\mathcal{G}_{n}(A,k)$ is Borel measurable.

\subsection{Measurability of The Function $c$}\label{appendix:cmeasurability}
In this subsection, we show the function $c:I^n\rightarrow  \mathbb{R}$ defined in Section \ref{subsubsection: pullback} is Borel measurable in $I^n$. Recall,
$$c(i_1,...,i_n)=\frac{1}{(n-1)!}\prod_{i\in I} n_i!$$
where $n_i=|\{k: i_k=i\}|$. 

Since there are finitely many sequence $(n_i)_{i\in I}$ such that, for all $i\in I$, $n_i\in\mathbb{N}$ and $\sum_{i\in I} n_i=n$, the range of the function $c$ is a finite set in $\mathbb{R}$. Consequently, to prove $c$ is Borel measurable, it is sufficient to show for any sequence of positive integers $m_1\ge m_2\ge ...\ge m_k\ge 1$ such that $m_1+...+m_k=n$, the set 
$$S_0(m_1,...m_k)=\{(\underbrace{i_1,...,i_1}_{m_1~\text{many}},\underbrace{i_2,...,i_2}_{m_2~\text{many}},...,\underbrace{i_k,...,i_k}_{m_k~\text{many}})\in I^n: i_1,...,i_k \text{~are~disjoint}\}$$
is Borel measurable. 

To see the sufficiency, we first define a set $S(m_1,...m_k)$ to be a subset of $I^n$ consisting of all permuted elements in the set $S_0(m_1,...m_k)$. Clearly, the cardinality of $S$ depends on the sequence $(m_1,...m_k)$. If $S_0(m_1,...m_k)$ is measurable, we have $S(m_1,...m_k)$, as a finite union of measurable sets, is measurable. Moreover, there are at most finitely many decreasing sequences of positive integers $(m_1',...,m_{l}')$ such that 
$m_1'!...m_l'!=m_1!...m_k!$ and $m_1'+...+m_l'=n$. Therefore, the set  $c^{-1}(\frac{m_1!...m_k!}{(n-1)!})\subset I^n$ is at most a finite union of the set $S(m_1',...,m_{l}')$ over all positive decreasing sequences $(m_1',...,m_{l}')$ such that 
$m_1'!...m_l'!=m_1!...m_k!$, thus is measurable.

Lastly, we prove, for any decreasing sequence $m_1\ge m_2\ge ...\ge m_k\ge 1$ such that $m_1+...+m_k=n$, $S_0(m_1,...,m_k)$ is measurable. Firstly, the set 
$$\{(\underbrace{i_1,...,i_1}_{m_1~\text{many}},\underbrace{i_2,...,i_2}_{m_2~\text{many}},...,\underbrace{i_k,...,i_k}_{m_k~\text{many}})\in I^n: i_1,...,i_k\in I\}$$
is Borel measurable as it is closed. In addition, $S_0(m_1,...,m_k)$ is obtained by subtracting a finite union of degenerate cases from this set, all of which are closed. Therefore, $S_0(m_1,...,m_k)$ is Borel measurable.

\subsection{Proof of Lemma 1} \label{appendix:shatprop}
In this subsection, we prove \hyperref[lemma:shatprop]{Lemma \ref{lemma:shatprop}}. 

Firstly, it is easy to prove that the set $K_n$ is closed in $I^{N!}/\sim_{N!}$ as it is a finite union of closed set. 

Moreover, for any permitted group size $n$, $\hat{s}_n$ is upper semi-continuous.
We prove by contradiction. Suppose there is a sequence such that $\hat{G}_k\rightarrow \hat{G}$ in $I^{N!}/\sim_{N!}$ such that $\limsup_{\hat{G}_k\rightarrow \hat{G}} \hat{s}_n(\hat{G}_k)>\hat{s}_n(\hat{G})$. Firstly, when $\hat{G}\notin K_n$, we have $\hat{G}_k\notin K_n$ for all large enough $k$. Therefore, $\limsup_{\hat{G}_k\rightarrow \hat{G}} \hat{s}_n(\hat{G}_k)=\hat{s}_n(\hat{G})=0$, which yields contradiction. Secondly, $\hat{G}\in K_n$. Then $\hat{G}_k\in K_n$ for infinitely many large k, since otherwise the limsup term will be zero. But $\hat{s}_n\circ P_n=s_n$ in $K_n$. By the upper semi-continuity of $s_n$, we have
$$\limsup_{\hat{G}_k\rightarrow \hat{G}} \hat{s}_n(\hat{G}_k) \le \hat{s}_n(\hat{G})$$
Contradiction. Recall that $\hat{s}$ is defined by
$$\hat{s}=N\max\left(\frac{1}{N'}\hat{s}_{N'},...,\frac{1}{N}\hat{s}_N\right)$$ 
Therefore, $\hat{s}$ is upper semi-continuous.

On the other hand, by \hyperref[A2]{Assumption (A2)}, there is a lower semi-continuous function $a\in L^1(I,\mu)$ such that
$$s_n([i_1,...,i_n])\le a(i_1)+...+a(i_n)$$
for all $N'\le n\le N$, $i_1,...,i_n\in I$. Consequently, 
$$\hat{s}_n(P_{n}^{-1}([i_1,...,i_n]))\le a(i_1)+...+a(i_n)$$
Moreover, since $s_n$ is non-negative, $a$ is non-negative everywhere. Therefore, for any $[i_1,...,i_{N!}]\in R_0$,
$$\hat{s}([i_1,...,i_{N!}])=0\le \sum_{k=1}^{N!} a(i_k)$$
For any  $[\underbrace{i_1,...,i_1}_{N!/n \text{~many}},...,\underbrace{i_n,...,i_n}_{N!/n \text{~many}}]\in R_n\subset K_n$, by the property of $R_n$, we have 
$$\hat{s}([i_1,...,i_1,...,i_n,...,i_n])=\frac{N}{n}\hat{s}_n([i_1,...,i_1,...,i_n,...,i_n])=\frac{N}{n}s_n([i_1,...,i_n])\le \frac{N}{n}\sum_{k=1}^{n} a(i_k)$$
Therefore, we define $\hat{a}=a/ (N-1)!$, we have,
$$\hat{s}([i_1,...,i_{N!}])\le \sum_{k=1}^{N!}\hat{a}(i_k)$$

\subsection{Proof of Lemma 3}\label{appendix: gamma}
In this subsection, we prove \hyperref[lemma:gamma]{Lemma \ref{lemma:gamma}}.
For any $S\subset I^n/\sim_n$, 
	\begin{align*}
	(Q_n)_{\#}\gamma_n(S)&=\gamma_n(Q_n^{-1}(S))\\
	&=\int_{Q_n^{-1}(S)} c d(Q_n^\#\tau_n)\\
	&=\sum_{\beta\in\mathcal{A}} \int_{Q_n^{-1}(S)\cap J_\beta} c d(Q_n^\#\tau_n)
	\end{align*}
	For each fixed $\beta$, by definition, 
	$J_\beta=J_{m,\sigma}$ for some $m\in\mathcal{M}$ and $\sigma\in S_n$. Therefore, $c=\frac{m_1!...m_k!}{(n-1)!}$ in $J_\beta$ and we have
	\begin{align*}
	\int_{Q_n^{-1}(S)\cap J_\beta} c d(Q_n^\#\tau_n)&= \frac{m_1!...m_k!}{(n-1)!}Q_n^\#\tau_n (Q_n^{-1}(S)\cap J_\beta)\\
	&=\frac{m_1!...m_k!}{(n-1)!}\tau_n(Q_n(Q_n^{-1}(S)\cap J_\beta))
	\end{align*}
	On the other hand, there are 
	$\frac{n!}{m_1!...m_k!}$ many $\alpha\in\mathrm{A}$ such that $Q_n(Q_n^{-1}(S)\cap J_\beta)=Q_n(Q_n^{-1}(S)\cap J_\alpha)$. We use  $[\beta]\subset \mathcal{A}$ to denote this collection of such indices. By definition, $\mathrm{A}$ can be partitioned into $|\mathcal{M}|$ components $\{[\beta_m]\}_{m\in\mathcal{M}}$ and
	$$\bigcup_{m\in\mathcal{M}} Q_n(Q_n^{-1}(S)\cap J_{\beta_m})=S$$
	where $\beta_m$ is a element in $[\beta_m]$. Therefore,
	\begin{align*}
	(Q_n)_{\#}\gamma_n(S)&=\sum_{\beta\in\mathcal{A}} \int_{Q_n^{-1}(S)\cap J_\beta} c d(Q_n^\#\tau_n)\\
	&=\sum_{m\in\mathcal{M}}\sum_{\beta\in[\beta_m]}\int_{Q_n^{-1}(S)\cap J_\beta} c d(Q_n^\#\tau_n)\\
	&=\sum_{m\in\mathcal{M}} \frac{n!}{m_1!...m_k!} \cdot \frac{m_1!...m_k!}{(n-1)!}\tau_n(Q_n(Q_n^{-1}(S)\cap J_{\beta_m}))\\
	&=\sum_{m\in\mathcal{M}} n \tau_n(Q_n(Q_n^{-1}(S)\cap J_{\beta_m}))\\
	&= n \tau_n(S)
	\end{align*}

\subsection{Proof of Lemma 4}\label{appendix:gammaprop}
	In this subsection, we prove \hyperref[lemma:gammaprop]{Lemma \ref{lemma:gammaprop}}. 
	
	Firstly, we prove the measure $\gamma_n$ is symmetric. For any permutation $\sigma\in S_n$, and any measurable sets $A_1,...,A_n\subset I$, we have
	\begin{align*}
	\gamma_n(A_{\sigma(1)}\times...\times A_{\sigma(n)})&= \int_{A_{\sigma(1)}\times...\times A_{\sigma(n)}} c d(Q_n^{\#}\tau_n)\\
	&=\sum_{\alpha\in\mathrm{A}} c(\alpha) Q_{n}^{\#}\tau_n(A_{\sigma(1)}\times...\times A_{\sigma(n)}\cap J_\alpha)\\
	&=\sum_{\alpha\in\mathrm{A}} c(\alpha) \tau_n(Q_n(A_{\sigma(1)}\times...\times A_{\sigma(n)}\cap J_\alpha))\\
	&=\sum_{\alpha\in\mathrm{A}} c(\alpha) \tau_n(Q_n(A_{1}\times...\times A_{n}\cap J_\alpha))\\
	&=\sum_{\alpha\in\mathrm{A}} c(\alpha) Q_{n}^{\#}\tau_n(A_{1}\times...\times A_{n}\cap J_\alpha)\\
	&=\gamma_n(A_{1}\times...\times A_{n})
	\end{align*}
	where $\{c(\alpha)\}_{\alpha\in\mathrm{A}}$ are a collection of constant determined by the function $c$ on the set $J_\alpha$.

	The second property will be proved by using the symmetric property of $\gamma_n$ and \hyperref[lemma:gamma]{Lemma \ref{lemma:gamma}}. By \hyperref[lemma:gamma]{Lemma \ref{lemma:gamma}}, we have 
	$$\sum_{k=0}^n k \tau_n(\mathcal{G}_n(A,k))=\sum_{k=0}^n \frac{k}{n} (Q_n)_{\#}\gamma_n(\mathcal{G}_n(A,k))=\sum_{k=0}^n \frac{k}{n} \gamma_n(Q_n^{-1}(\mathcal{G}_n(A,k)))$$
	In particular,
	$Q_n^{-1}(\mathcal{G}_n(A,k))$ is the union of Cartesian products of $k$ sets $A$ and $n-k$ sets $A^c$. There are $\frac{n!}{k!(n-k)!}$ many such Cartesian products. Therefore, since $\gamma_n$ is symmetric, 
	$$\gamma_n(Q_n^{-1}(\mathcal{G}_n(A,k)))=\frac{n!}{k!(n-k)!}\gamma_n(\underbrace{A\times...\times A}_{k~\text{many}}\times \underbrace{A^c\times...\times A^c}_{n-k~\text{many}})$$
	On the other hand, there are $\frac{(n-1)!}{(k-1)!(n-k)!}$  Cartesian products of $k$ sets $A$ and $n-k$ sets $A^c$ in the form $A\times S_2\times ...\times S_n$ where $S_j\in \{A,A^c\}$. We use $\mathcal{S}_k$ to denote the collection of such Cartesian products. i.e.
	$$\mathcal{S}_k=\{A\times S_2\times ...\times S_n: S_j\in \{A,A^c\}, |\{j: S_j=A\}|=k-1\}$$
	Therefore, again by the symmetry of the measure $\gamma_n$, 
	$$\gamma_n(\cup_{S\in\mathcal{S}_k}S )=\frac{(n-1)!}{(k-1)!(n-k)!}\gamma_n(\underbrace{A\times...\times A}_{k~\text{many}}\times \underbrace{A^c\times...\times A^c}_{n-k~\text{many}})$$
	Consequently,
	$$\sum_{k=0}^n k \tau_n(\mathcal{G}_n(A,k))=\sum_{k=0}^n \frac{k}{n} \gamma_n(Q_n^{-1}(\mathcal{G}_n(A,k)))=\sum_{k=1}^n\gamma_n(\cup_{S\in\mathcal{S}_k}S)=\gamma_n(\cup_{S\in\cup_{k=1}^n\mathcal{S}_k}S)$$
	Note that,
	$$\cup_{k=1}^n\mathcal{S}_k=\{A\times S_2\times ...\times S_n: S_j\in \{A,A^c\}\}$$
	Thus, 
	$$\cup_{S\in\cup_{k=1}^n\mathcal{S}_k}S=A\times I\times...\times I$$
	In conclusion,
	$$\sum_{k=0}^n k \tau_n(\mathcal{G}_n(A,k))=\gamma_n(A\times I\times...\times I)$$

\subsection{Proof of Lemma 5} \label{appendix:duality}
In this subsection, we prove \hyperref[lemma: duality]{Lemma \ref{lemma: duality}}.  We start by stating and proving the following lemma:
\begin{lemma}\label{lemma: symmetricS}
	If $S$ is a symmetric function,
	$$\sup_{\hat{\gamma}\in \hat{\Gamma}_{sym}} \int_{I^{N!}} S d\hat{\gamma}=\sup_{\hat{\gamma}\in \hat{\Gamma}} \int_{I^{N!}} S d\hat{\gamma}$$ 
\end{lemma}

\begin{proof}
	Since $\hat{\Gamma}_{sym}\subset\hat{\Gamma} $, we have 
	$$\sup_{\hat{\gamma}\in \hat{\Gamma}_{sym}} \int_{I^{N!}} S d\hat{\gamma}\le\sup_{\hat{\gamma}\in \hat{\Gamma}}\int_{I^{N!}} S d\hat{\gamma}$$ 
	Conversely, for any $\gamma_0\in \hat{\Gamma}$, we define
	$$\hat{\gamma}=\dfrac{1}{(N!)!}\sum_{\sigma\in S_{N!}}(f_\sigma)_\#\gamma_0$$
	where $f_{\sigma}: I^{N!}\rightarrow I^{N!}$ is defined by
	$f_{\sigma}(i_1,...,i_{N!})=(i_{\sigma(1)},...,i_{\sigma(N!)})$.

	We claim $\hat{\gamma}$ is in $\hat{\Gamma}_{sym}$ and
	$$\int_{I^{N!}} S d\hat{\gamma}= \int_{I^{N!}} S d\gamma_0$$ 
	We omit the routine work to check $\hat{\gamma}$ is in $\hat{\Gamma}_{sym}$. Since $S$ is symmetric, ${\gamma}$ induces the same total welfare as $\gamma$:
	\begin{align*}
	\int_{I^{N!}} S d\hat{\gamma}=& \dfrac{1}{(N!)!}\sum_{\sigma\in S_{N!}}\int_{I^{N!}} S d(f_\sigma)_\#\gamma_0=\dfrac{1}{(N!)!}\sum_{\sigma\in S_{N!}}\int_{I^{N!}} s\circ f_\sigma d\gamma_0\\
	= & \frac{1}{(N!)!} \sum_{\sigma\in S_{N!}}\int_{I^{N!}} S d\gamma_0=\int_{I^{N!}} S d\gamma_0
	\end{align*}
	Since $\gamma_0\in\hat{\Gamma}$ is arbitrary,
	$$\sup_{\hat{\gamma}\in \hat{\Gamma}_{sym}} \int_{I^{N!}} S d\hat{\gamma}\ge\sup_{\hat{\gamma}\in \hat{\Gamma}}\int_{I^{N!}} S d\hat{\gamma}$$ 
\end{proof}

By the duality theorem for multi-marginal transport problem \cite{Kellerer84}, we have
$$\sup_{\hat{\gamma}\in \hat{\Gamma}} \int_{I^{N!}} \frac{S}{N} d\hat{\gamma}= \inf_{(u_j)_{j=1}^{N!}\in \tilde{\mathcal{U}}} \sum_{j=1}^{N!} \int_I u_j d\mu$$
where 
$$\tilde{\mathcal{U}}=\left\{(u_1,..., u_{N!}) \in (L^1(I,\mu))^{N!}: \sum_{j=1}^{N!} u_j(i_j)\ge \frac{S(i_1,...,i_{N!})}{N}, \forall i_1,i_2,...,i_{N!}\in I \right\}$$ and the infimum could be achieved. Moreover, by \hyperref[lemma: symmetricS]{Lemma \ref{lemma: symmetricS}},
$$\sup_{\hat{\gamma}\in \hat{\Gamma}_{sym}} \int_{I^{N!}} \frac{S}{N} d\hat{\gamma}=\sup_{\hat{\gamma}\in \hat{\Gamma}} \int_{I^{N!}} \frac{S}{N} d\hat{\gamma}$$ 

So it remains to show 
\begin{equation}\label{eqn: syminf}
\inf_{(u_j)_{j=1}^{N!}\in \tilde{\mathcal{U}}} \sum_{j=1}^{N!} \int_I u_j d\mu=\inf_{u\in \hat{\mathcal{U}}} \int_I ud\mu
\end{equation}
and the infimum on the right hand side could be achieved.

Firstly, take any $u\in\hat{\mathcal{U}}$, we have 
$(u,u,...,u)\in \tilde{\mathcal{U}}$. Therefore,
$$\inf_{(u_j)_{j=1}^{N!}\in \tilde{\mathcal{U}}} \sum_{j=1}^{N!} \int_I u_j d\mu\le\inf_{u\in \hat{\mathcal{U}}} \int_I ud\mu$$
Conversely, for any minimizer $(u_j^*)_{j=1}^{N!}\in \tilde{\mathcal{U}}$ solving the left hand side of \hyperref[eqn: syminf]{Equation \ref{eqn: syminf}}, we define 
\begin{equation}\label{eqn: appendixu}
u=\sum_{n=1}^{N!} u_n^*
\end{equation}
Then, for any $i_1,...,i_{N!}\in I$, we have, 
\begin{align*}
\sum_{j=1}^{N!} u(i_j)& =\sum_{j=1}^{N!}\sum_{n=1}^{N!} u_n^*(i_j) =\sum_{j=1}^{N!}\left[ \sum_{k=0}^{{N!}-1} u_j^*(i_{j+k})\right]\\
& \ge \frac{1}{N}\left[S(i_1,..., i_{N!}) + S(i_2,..., i_{N!}, i_1)+ ... + S(i_{N!}, i_1, ...., i_{N!-1})\right] \\
& = (N-1)!S(i_1,...,i_{N!})
\end{align*}
where $i_{N!+k}$ is defined to be $i_{k}$ for $1\le k\le N!$ and the last equality is implied by the symmetry of $S$. Moreover, as a finite sum of upper semi-continous integrable functions, $u\in L^1(I,\mu)$.  Thus, $u\in \mathcal{U}$, and we have  
$$\inf_{(u_j)_{j=1}^{N!}\in \tilde{\mathcal{U}}} \sum_{j=1}^{N!} \int_I u_j d\mu\ge\inf_{u\in \hat{\mathcal{U}}} \int_I ud\mu$$
In conclusion,  
$$\inf_{(u_j)_{j=1}^{N!}\in \tilde{\mathcal{U}}} \sum_{j=1}^{N!} \int_I u_j d\mu=\inf_{u\in \hat{\mathcal{U}}} \int_I ud\mu$$
and $u=\sum_{n=1}^{N!} u_n^*$ solves the minimization problem on the right hand side of \hyperref[eqn: syminf]{Equation \ref{eqn: syminf}}.

\subsection{Proof of Lemma 6} \label{appendix: Uhat }
In this subsection, we prove \hyperref[lemma: Uhat]{Lemma \ref{lemma: Uhat}}.
Firstly, take any $u\in \hat{\mathcal{U}}$ and any $[i_1,...,i_1,i_2,...,i_2,...,i_n,...,i_n]\in K_n$, We have
\begin{align*}
\sum_{k=1}^{n} \frac{N!}{n} u(i_k)\ge &(N-1)!\hat{s}([i_1,...,i_1,i_2,...,i_2,...,i_n,...,i_n])\\
\ge &\frac{N!}{n}\hat{s}_n([i_1,...,i_1,i_2,...,i_2,...,i_n,...,i_n])\\
=&\frac{N!}{n}s_n([i_1,...,i_n])
\end{align*}
Therefore, $u\in \mathcal{U}_n$. Since $n$ is arbitrary, $u\in \mathcal{U}=\cap_{n=N'}^N\mathcal{U}_n$. Conversely, take any $u\in \mathcal{U}$, we check $u\in \hat{\mathcal{U}}$ region by region. Firstly, take $[i_1,...,i_{N!}]\in R_0$, it is clear that $$\sum_{k=1}^{N!} u(i_k)\ge 0= (N-1)!\hat{s}([i_1,...,i_{N!}])$$
Moreover, for any permitted group size $N'\le n\le N$, we take $[i_1,...,i_1,i_2,...,i_2,...,i_n,...,i_n]\in R_n$,
\begin{align*}
\sum_{k=1}^n \frac{N!}{n} u(i_k)&= \frac{N!}{n} \sum_{k=1}^n  u(i_k)\\
&\ge \frac{N!}{n}   s_n([i_1,...,i_n])\\
&= \frac{N!}{n}  \hat{s}_n([i_1,...,i_1,i_2,...,i_2,...,i_n,...,i_n])\\
&= (N-1)! \hat{s}([i_1,...,i_1,i_2,...,i_2,...,i_n,...,i_n])
\end{align*}
Since $\{R_0,R_{N'},...,R_{N}\}$ is a partition of $I^{N!}/\sim_{N!}$, we have $u\in \hat{ \mathcal{U}}$.

\subsection{Proof of Theorem 2} \label{appendix: ctsmin}

Now, we finish the \hyperref[proof: duality]{proof of duality theorem} by proving there is a continuous minimizer for the welfare minimization problem
\begin{equation} \label{eqn: min}
\inf_{u\in {\mathcal{U}}} \int_{I} u d\mu
\end{equation}
where $\mathcal{U}=\{u\in L^1(I,\mu): \sum_{i\in G} u(i)\ge s(G), \forall G\in \mathcal{G}\}$.

Firstly, by \hyperref[lemma: duality]{Lemma \ref{lemma: duality}} and \hyperref[lemma: Uhat]{Lemma \ref{lemma: Uhat}}, there is a minimizer $u^*\in L^1(I,\mu)$ minimizing \hyperref[eqn: min]{Equation \ref{eqn: min}}. On the other hand, by \hyperref[lemma: maxachieved]{Lemma \ref{lemma: maxachieved}}, there is a maximizer $\tau^*$ for the welfare maximization problem. Define $\mu_n=\tau^*(\mathcal{G}_n)$. That is, $\mu_n$ is the type distribution of players who are assigned to some $n$-person group under assignment $\tau^*$. Clearly, $\sum_{n=N'}^N \mu_n=\mu$.

We claim that, for any permitted group size $n$, there is a continuous function $v_n\in C(I)\cap \mathcal{U}_n$ such that $v_n\le u^*$ on $I$ and $\int_I u^* d\mu_n=\int_I v_n d\mu_n$. We note, if the claim is proved to be true, we can take $v=\max(v_{N'},...,v_N)$. It is clear, as the maximum of finitely many continuous functions, $v$ is continuous. Moreover, as $v_n\in  \mathcal{U}_n$, $v$ is in $\mathcal{U}$. By $v\le u^*$,  $v$ is also a minimizer of \hyperref[eqn: min]{Equation \ref{eqn: min}}.

Now, we fixed a permitted group size $n$ and prove the claim. Our proof is based on the trick used in the proof of Theorem 4.1.1 in \cite{pass2011structural}. This trick can also be found in \cite{gangbo1998optimal} and \cite{carlier2008optimal}.

We define $\gamma^*_n$ by $\tau^*_n$ using the change of variable trick in \hyperref[eqn: defgamma]{Equation \ref{eqn: defgamma}}. By \hyperref[lemma:gamma]{Lemma \ref{lemma:gamma}},
$$\int_{I^n} \frac{\tilde{s}_n}{n} d\gamma^*_n=\int_{\mathcal{G}_n} {s}_n d\tau^*_n$$
where $\tilde{s}_n=s\circ Q_n$ is a continuous function on $I^n$. Moreover, by the proof of equivalence lemma in \hyperref[appendix: equivalencerelation]{Appendix \ref{appendix: equivalencerelation}}, $\sum_{i\in G} u^*(i)=s_n(G)$ for $\tau_n^*$-almost all $G\in \mathcal{G}_n$. Therefore, 
$$\int_{I^n} \tilde{s}_n d\gamma^*_n=n\int_I u^* d\mu_n$$

We define $n$ functions $(w_1,...,w_n)$ as follows:
$$w_1(i_1)=\sup_{i_k\in I, k\ge 2}\left(\tilde{s}_n(i_1,...,i_k)- \sum_{k=2}^n u^*(i_k)  \right)$$
Inductively, for $m\ge 2$, 
$$w_m(i_m)=\sup_{i_k\in I, k\neq m}\left(\tilde{s}_n(i_1,...,i_k)- \sum_{k=1}^{m-1} w_k(i_k) -\sum_{k=m+1}^n u^*(i_k)  \right)$$
Inductively, by $\sum_{k=1}^{n} u^*(i_k)\ge \tilde{s}_n(i_1,...,i_k)$ for all $i_1,...,i_k\in I$, we have 
\begin{equation}\label{eqn: wleu}
u^*(x_k)\ge w_k(x_k)
\end{equation}
for all $1\le k\le n$, $x_k\in I$. Moreover, recall,
$$w_n(i_n)=\sup_{i_k\in I, k\neq n}\left(\tilde{s}_n(i_1,...,i_n)- \sum_{k=1}^{n-1} w_k(i_k) \right)$$
we have for all $1\le m\le n$,
$$w_m(i_m)\ge \sup_{i_k\in I, k\neq m}\left(\tilde{s}_n(i_1,...,i_n)- \sum_{k\neq m} w_k(i_k) \right)$$
The definition of $w_{m-1}$ implies for all $i_1,...,i_n\in I$,
$$u^*(i_m)\ge s_n(i_1,...,i_n)-\sum_{k=1}^{m-1} u^*(i_k)-\sum_{k=m+1}^{n} w_k(i_k)$$
Therefore,
\begin{align*}
w_m(x_m) &= \sup_{i_k\in I, k\neq m}\left(\tilde{s}_n(i_1,...,i_k)- \sum_{k=1}^{m-1} w_k(i_k) -\sum_{k=m+1}^n u^*(i_k)  \right)\\
& \le \sup_{i_k\in I, k\neq m}\left(\tilde{s}_n(i_1,...,i_n)- \sum_{k\neq m} w_k(i_k) \right)
\end{align*}
Consequently, for all $1\le m\le n$ and all $i_1,...,i_k\in I$,
$$w_m(i_m)=\sup_{i_k\in I, k\neq m}\left(\tilde{s}_n(i_1,...,i_n)- \sum_{k\neq m} w_k(i_k) \right)$$
That is, $(w_1,...,w_n)$ are $\tilde{s}_n$-conjugate. 
Since $\tilde{s}_n$ is continuous and $I$ is compact, $w_1,...,w_n$ are continuous. %Santambrogio15 Page 11
Define $v_n=\frac{1}{n}(w_1+...+w_n)\in C(I)$. By \hyperref[eqn: wleu]{Equation \ref{eqn: wleu}}, $v_n\le u^*$. By the symmetry of $\tilde{s}_n$,  $v_n\in \mathcal{U}_n$. Moreover, as all marginals of $\gamma_n^*$ are $\mu_n$,
$$n\int_I u^* d\mu_n\ge\sum_{k=1}^n \int_I w_k d\mu_n
\ge \int_{I^n} \tilde{s}_n d\gamma_n^*$$
But $\int_{I^n} \tilde{s}_n d\gamma_n^*=n\int_I u^* d\mu_n$. Thus, $\int_I u^* d\mu_n= \int_I v_n d\mu_n$. Therefore, the claim is proved.

\subsection{Proof of Lemma 7}\label{appendix: maxachieved}
In this subsection, we prove \hyperref[lemma: maxachieved]{Lemma \ref{lemma: maxachieved}}. By \hyperref[prop:linearprogrammingtotransport]{Proposition \ref{prop:linearprogrammingtotransport}} and \hyperref[lemma: symmetricS]{Lemma \ref{lemma: symmetricS}}, we have 
	$$\sup_{\tau\in \Tau} \sum_{n=N'}^N \int_{\mathcal{G}_n} s_n d\tau_n=\sup_{\hat{\gamma}\in \hat{\Gamma}_{sym}} \int_{I^{N!}} \frac{S}{N} d\hat{\gamma}=\sup_{\hat{\gamma}\in \hat{\Gamma}} \int_{I^{N!}} \frac{S}{N} d\hat{\gamma}$$
	
	We first show there is a maximizer $\hat{\gamma}$ solving the right maximization problem, then we construct a maximizer $\tau$ for the left maximization problem.
	
	To show the existence of a maximizer $\hat{\gamma}$, we show $\hat{\Gamma}$ is compact and the map $\hat{\gamma}\rightarrow \int_{I^{N!}} S d\hat{\gamma}$ is upper semi-continuous. 

	\begin{lemma}\label{lemma: compact}
		$\hat{\Gamma}$ is compact in the weak-* topology.
	\end{lemma}
	\begin{proof}
	 Since $I$ is Polish space, by Ulam's theorem, $\{\mu/\norm{\mu}\}$ is tight in $I$. i.e. for any $\varepsilon>0$, there exists a compact set $K_\varepsilon$, such that $\mu(I-K_\varepsilon)<\varepsilon$. Hence, $\hat{\Gamma}$ is uniformly tight as,  for any $\hat{\gamma}\in \hat{\Gamma}$,
	$$\hat{\gamma}(I^{N!}-K_\epsilon^{N!})\le N!\mu(I-K_\varepsilon)\le N!\varepsilon $$
	Moreover, $\hat{\Gamma}$ is uniformly bounded as, for any $\hat{\gamma}\in \hat{\Gamma}$, $\norm{\hat{\gamma}}=\norm{\mu}$.
	By Prokhorov's theorem (Theorem 8.6.2 in \cite{Bogachev2}), $\hat{\Gamma}$ has compact closure in weak* topology. 
	
	Thus, to show $\hat{\Gamma}$ is compact, it is sufficient to show it is closed: for any convergent sequence $(\hat{\gamma}_k)_{k\in\mathbb{N}}$ in $\hat{\Gamma}$ such that $\hat{\gamma}_k\rightarrow \hat{\gamma}$, we have, for all measurable $S\subset I^{N!}$,
	$$\hat{\gamma}(S)=\lim_{k\rightarrow \infty} \hat{\gamma}_k(S)$$
	For any measurable set $A\subset I$, by taking $S$ to be the Cartesian product of one set A and $N!-1$ sets $I$, we can show all the marginals of $\hat{\gamma}$ is $\mu$. Therefore, $\hat{\gamma}\in \hat{\Gamma}$. i.e. $\hat{\Gamma}$ is closed.
\end{proof}

\begin{lemma}\label{lemma: usc}
The functional $\hat{\gamma}\rightarrow \int_{I^{N!}} S d\hat{\gamma}$ is upper semi-continuous in the weak-* topology.
\end{lemma}
\begin{proof}
Firstly, we claim it is without loss of generality to assume $S$ is non-positive: by \hyperref[eqn: bigSbdd]{Equation \ref{eqn: bigSbdd}}, $S$ is bounded from above by the function lower-semi continuous and integrable $A:I^{N!}\rightarrow \mathbb{R}$ where
$$ A(i_1,...,i_{N!})= \sum_{j=1}^{N!} \hat{a}(i_j)$$
We can study $s-A$, which is a non-positive upper semi-continuous and integrable function. By the non-positive upper semi-continuity, there is a decreasing sequence of $S_l$ converging to $S$ pointwisely, where $S_l$ is a continuous bounded function on $I^{N!}$. Taking $\hat{\gamma}_k\rightarrow \hat{\gamma}$, by monotone convergence theorem, we have
$$\int_{I^{N!}} S d\hat{\gamma}=\lim_{l\rightarrow \infty}\int_{I^{N!}} S_l d\hat{\gamma}=\lim_{l\rightarrow \infty}\lim_{k\rightarrow \infty} \int_{I^{N!}} S_ld\hat{\gamma}_k\ge \limsup_{k\rightarrow\infty}\int_{I^{N!}} Sd\hat{\gamma}_k$$
\end{proof} 

Therefore, by \hyperref[lemma: compact]{Lemma \ref{lemma: compact}} and \hyperref[lemma: usc]{Lemma \ref{lemma: usc}}, there is a $\hat{ \gamma}^*\in\hat{ \Gamma}$ solves the maximization problem
$$\sup_{\hat{\gamma}\in \hat{\Gamma}} \int_{I^{N!}} \frac{S}{N} d\hat{\gamma}$$
By the construction in \hyperref[lemma: symmetricS]{Lemma \ref{lemma: symmetricS}}, we have a maximizer $\hat{\gamma}^*_{sym}\in \hat{ \Gamma}_{sym}$. Then we repeat the construction in \hyperref[eqn: constructtau]{Equation \ref{eqn: constructtau}}, and obtain a  $\tau^*=(\tau^*_{N'},...,\tau^*_{N})\in \Tau$ such that
$$\sum_{n=N'}^N \int_{\mathcal{G}_n} s_n d\tau_n^*=\int_{I^{N!}} \frac{S}{N} d\hat{\gamma}^*_{sym}=\sup_{\hat{\gamma}\in \hat{\Gamma}_{sym}} \int_{I^{N!}} \frac{S}{N} d\hat{\gamma}=\sup_{\tau\in \Tau} \sum_{n=N'}^N \int_{\mathcal{G}_n} s_n d\tau_n$$
i.e. $\tau^*$ solves the maximization problem
$$\sup_{\tau\in \Tau} \sum_{n=N'}^N \int_{\mathcal{G}_n} s_n d\tau_n$$

\subsection{Proof of Lemma 8} \label{appendix: equivalencerelation}

In this subsection, we prove \hyperref[lemma: equivalencerelation]{Lemma \ref{lemma: equivalencerelation}}. Given any maximizer $\tau^*$ and minimizer $u^*$ in the duality equation \hyperref[eqn: duality]{Equation \ref{eqn: duality}}, by $u^*\in \mathcal{U}$, we have the no-blocking condition:
$$\sum_{i\in G} u^*(i)\ge s(G), \forall G\in \mathcal{G}$$
For every permitted group size $n$, we define $\gamma_n^*$ by $\tau_n^*$ by \hyperref[eqn: defgamma]{Equation \ref{eqn: defgamma}}. By \hyperref[lemma:gamma]{Lemma \ref{lemma:gamma}}, we have
$$\sum_{n=N'}^N \int_{\mathcal{G}_n} \left(u^*(i_1)+...+u^*(i_n)\right) d\tau_n^*=\sum_{n=N'}^N \int_{I^n} \frac{u^*(i_1)+...+u^*(i_n)}{n} d\gamma_n^*
$$
Moreover, by \hyperref[lemma:gammaprop]{Lemma \ref{lemma:gammaprop}}, $\gamma^*\in\Gamma$. Therefore, we have
$$\sum_{n=N'}^N \int_{I^n} \frac{u^*(i_1)+...+u^*(i_n)}{n} d\gamma_n^*=\int_I u^* d\mu$$
Consequently,
$$\int_I u^* d\mu=\sum_{n=N'}^N \int_{\mathcal{G}_n} \left(u^*(i_1)+...+u^*(i_n)\right) d\tau_n^*\ge \sum_{n=N'}^N \int_{\mathcal{G}_n} s_n d\tau_n^*$$
By the duality relation, \hyperref[eqn: duality]{Equation \ref{eqn: duality}}, 
$$\int_I u^* d\mu= \sum_{n=N'}^N \int_{\mathcal{G}_n} s_n d\tau_n^*$$
Therefore, $\sum_{i\in G} u^*(i)\le s(G)$ for $\tau^*$-almost all $G\in \supp(\tau^*)$. By the continuity of $u$ and $s$, we have the feasibility condition
$$\sum_{i\in G} u^*(i)\le s(G), \forall G\in \supp(\tau^*)$$
That is, $\tau^*$ is a stable assignment with imputation given by $u^*$.

Conversely, given a stable assignment $\tau^*$ with an imputation $u^*$, we show $\tau^*$ solves the maximization problem and $u^*$ solves the minimization problem. Firstly, since $u^*\in \mathcal{U}$ is an imputation such that there is no blocking coalitions, for any assignment $\tau\in \Tau$, we have
$$\sum_{n=N'}^N \int_{\mathcal{G}_n} s_n d\tau_n\le \int_I u^* d\mu$$
However, since $\tau^*$ is stable, by feasibility condition,
$$\sum_{n=N'}^N \int_{\mathcal{G}_n} s_n d\tau_n^*\ge \int_I u^* d\mu$$
Therefore, $\tau^*$ solves the maximization problem in \hyperref[eqn: duality]{Equation \ref{eqn: duality}}.
Secondly, by the no-blocking condition of stable assignments, for any imputation $u\in \mathcal{U}$, we have
$$\sum_{n=N'}^N \int_{\mathcal{G}_n} s_n d\tau_n^*\le \int_I u d\mu$$
Again by feasibility condition, we have 
$$\sum_{n=N'}^N \int_{\mathcal{G}_n} s_n d\tau_n^*\ge \int_I u^* d\mu$$
Therefore, $u^*$ solves the minimization problem in \hyperref[eqn: duality]{Equation \ref{eqn: duality}}.

\section{Omitted Proofs in Section 3} \label{appendix: sect3}

In this section, we prove \hyperref[theorem: infdimexistence]{Theorem \ref{theorem: infdimexistence}} and \hyperref[theorem: infdimduality]{Theorem \ref{theorem: infdimduality}}. We use the same three-step procedure introduced \cite{gretsky1992nonatomic} to show the existence of a stable assignment. Firstly, we study the optimization problems and show the optimizers exist. Secondly, we prove the duality property to link the assignment and the imputation.  Lastly, we prove the equivalence between the solutions of optimization problems and the set of stable assignments.

\begin{lemma}
	$\Gamma_{\mathcal{G}}$ is a compact set in $\mathcal{M}_+(\mathcal{G})$.
\end{lemma}

\begin{proof}
	Since $I$ is a Polish space,  $\mathcal{M}_+(I)$, the set of nonnegative Borel measures on $I$, with weak topology on it is a Polish space by Theorem 8.9.4 in \cite{Bogachev2}. By the same argument, $\mathcal{M}_+(\mathcal{M}_+(I))$ is a Polish space.
	
	On the other hand, since $I$ is a Polish space and $\mu$ is a probability measure, by Ulam's tightness theorem, $\{\mu\}\subset \mathcal{M}(I)$ is tight. That is, for any $\delta>0$, there is a compact set $K_\delta\subset I$ such that $\mu(I-K_\delta)<\delta$.

	Next, we show $\mathcal{G}$, as a subset of $\mathcal{M}_+(I)$, is compact. Firstly, $\mathcal{G}$ is uniformly tight. Since for any $\delta>0$, there is a compact set $K_\delta\subset I$ such that $|\nu|(I-K_\delta)\le \mu(I-K_\delta)<\delta$, for all $\nu\in\mathcal{G}$. Moreover, $\mathcal{G}$ is uniformly bounded in variation norm as $\norm{\nu}\le \varepsilon$ for all $\nu\in\mathcal{G}$. By Prohorov's theorem for Borel measures (Theorem 8.6.2 in \cite{Bogachev2}), $\mathcal{G}$ has compact closure.  It is routine to check $\mathcal{G}$ is closed. Therefore, $\mathcal{G}$ is compact.
	
	Furthermore, we show $\Gamma_{\mathcal{G}}$ is compact in $\mathcal{M}_+(\mathcal{M}_+(I))$. Firstly, $\Gamma_{\mathcal{G}}$ is uniformly tight since $\gamma(\mathcal{M}_+(I)-\mathcal{G})=0$ for all $\gamma\in \Gamma_{\mathcal{G}}$. Moreover, the total variation norm of $\gamma$ is uniformly bounded as for all $\gamma\in \Gamma_{\mathcal{G}}$
	$$\gamma(\mathcal{M}_+(I))=\gamma(\mathcal{G})\le \frac{1}{\varepsilon'}\int_{\mathcal{G}} \norm{\nu} d\gamma=\gamma(\mathcal{G})\le \frac{1}{\varepsilon'}\int_{\mathcal{G}} \nu(I) d\gamma=\frac{\mu(I)}{\varepsilon'}=\frac{\norm{\mu}}{\varepsilon'} $$
	Again by Prohorov's theorem for Borel measures, $\mathcal{G}$ has compact closure in $\mathcal{M}_+(\mathcal{M}_+(I))$. Moreover, since the map $\nu\rightarrow \nu(A)$ is continuous and bounded on $\mathcal{M}_+(I)$ for any Borel set $A\subset I$, for any convergent sequence $\gamma_n$ in $\Gamma_{\mathcal{G}}$ and any Borel set $A\subset I$,
	$$\mu(A)=\lim_{n\rightarrow \infty}\int_\mathcal{G} \nu(A) d\gamma_n= \int_\mathcal{G} \nu d\gamma$$
	Therefore, $\Gamma_{\mathcal{G}}$ is closed. Hence, 	$\Gamma_{\mathcal{G}}$ is compact.
\end{proof}

\begin{lemma}
	For the cooperative game $((I,\mu),s, \varepsilon', \varepsilon)$  satisfying \hyperref[B1]{Assumption (B1)} and \hyperref[B2]{Assumption (B2)}, there is an assignment $\gamma\in\Gamma_{\mathcal{G}}$ solving the welfare maximization problem
	$$\sup_{\gamma\in \Gamma_{\mathcal{G}}} \int_{\mathcal{G}} s d\gamma$$
\end{lemma}

\begin{proof}
	
	Firstly, we prove the map $\gamma\rightarrow \int_\mathcal{G} s d\gamma$ is upper-semi continuous on $\mathcal{M}_+(\mathcal{M}_+(I))$ if $s$ is non-positive. Since $s$ is a upper semi-continuous function and has a uniform bound, we take a decreasing sequence of $s_l$ converging to s pointwisely, where $s_l$ is continuous. Let $\gamma_n\rightarrow \gamma$ in $\mathcal{M}_+(\mathcal{G})$, by the monotone convergence theorem and the fact $s_l$ is decreasing, we have
	$$\int_\mathcal{G} s d\gamma=\lim_{l\rightarrow \infty}\int_\mathcal{G} s_ld\gamma=\lim_{l\rightarrow \infty}\lim_{n\rightarrow \infty} \int_\mathcal{G} s_ld\gamma_n\ge \limsup_{n\rightarrow\infty} \int_\mathcal{G} s d\gamma_n$$
	
	Moreover, define a function $h:\mathcal{G}\rightarrow \mathbb{R}$ by $h(\nu)=\int_I a d\nu$. By \hyperref[B2]{Assumption (B2)}, $s-h$ is non-positive in $\mathcal{G}$. Therefore, taking a sequence of $\gamma_n\in\Gamma_{\mathcal{G}}$ approaching the maximum, by compactness of the choice set $\Gamma_{\mathcal{G}}$, a subsequence of it converges to some $\gamma\in \Gamma_{\mathcal{G}}$. By selecting the subsequence, we suppose $\gamma_n\rightarrow \gamma$ in $\mathcal{M}_+(\mathcal{G})$. Therefore,
	$$\int_\mathcal{G} s d\gamma-\int_\mathcal{G} h d\gamma \ge \limsup_{n\rightarrow\infty} \left( \int_\mathcal{G} s d\gamma_n- \int_\mathcal{G} h d\gamma_n\right)$$
	However, we note h is continuous and bounded on $\mathcal{G}$ since
	$h(\nu)=\int_I a d\nu\le \norm{a}_{\infty}\norm{\nu}_0 $ and $h(\nu)\le \norm{a}_{\infty} \varepsilon $. Therefore, by the definition of weak convergence on in $\mathcal{M}_+(\mathcal{G})$, we have
	$\lim_{n\rightarrow \infty}\int_\mathcal{G} h d\gamma_n=\int_\mathcal{G} h d\gamma$. Therefore, 
	$$\int_\mathcal{G} s d\gamma \ge \limsup_{n\rightarrow\infty} \int_\mathcal{G} s d\gamma_n$$
	i.e. $\gamma$ is a maximizer.
\end{proof}

\begin{lemma}
	For the game $((I,\mu),s, \varepsilon', \varepsilon)$ satisfying \hyperref[B1]{Assumption (B1)} and \hyperref[B2]{Assumption (B2)}, there is an imputation $u\in \mathcal{U}$ solving the payoff minimization problem  
	$$\inf_{u\in \mathcal{U}} \int_{I} u d\mu$$
\end{lemma}

\begin{proof}
	By By \hyperref[B2]{Assumption (B2)}, there is a function $a\in L^1(I,\mu)$ such that $\int_I ad\nu\ge s(\nu)$ for all $\nu\in\mathcal{G}$. Therefore, $a\in \mathcal{U}$.
	Note that for any $u\in \mathcal{U}$, $\min(u,a)\in \mathcal{U}$. Therefore, we define a subset $\mathcal{U}_a\subset \mathcal{U}$ where
	$$\mathcal{U}_a=\{u\in\mathcal{U}: 0\le u\le \norm{a}_\infty \}$$
	Clearly, $\inf_{u\in \mathcal{U}} \int_{I} u d\mu=\inf_{u\in \mathcal{U}_a} \int_{I} u d\mu$ as all $u\in \mathcal{U}$ is pointwisely non-negative.
	
	Next, we define a sequence of spaces $(\mathcal{U}_a(k))_{k\in \mathbb{N}}$ in  $L^1(\mu)$ where
	$$\mathcal{U}_a(k)=\left\{u\in \mathcal{U}_a: \int_I u d\mu\le \inf_{u\in \mathcal{U}} \int_{I} u d\mu+1/k  \right\}$$
	We note for any $k\in \mathbb{N}$, $\mathcal{U}_a(k)$ is non-empty. Moreover, it is uniformly integrable: for any $c>0$, we take $\delta=c/\norm{a}_\infty>0$. Then, for any Borel measurable $E\subset I$ with $\mu(E)<\delta$, $\int_{E} u d\mu\le c$. By Dunford-Pettis theorem (Theorem 3 in \cite{diestel1978vector}), the weak closure of $\mathcal{U}_a(k)$ is weakly compact in $L^1(\mu)$. But it is easy to see $\mathcal{U}_a(k)$ is closed in weak topology. Therefore, 	$\mathcal{U}_a(k)$ is non-empty, closed and compact in $L^1(\mu)$ endowed with weak topology. Note the sequence $\mathcal{U}_a(k)$ is decreasing, by Cantor's intersection theorem,
	$$\bigcap_{k\in \mathbb{N}} \mathcal{U}_a(k)\neq \emptyset$$
	The element $u\in \bigcap_{k\in \mathbb{N}} \mathcal{U}_a(k)$ has the property that $\int_I u d\mu\le \inf_I u d\mu$.	
\end{proof}

Next, we prove the \hyperref[theorem: infdimduality]{Theorem \ref{theorem: infdimduality}}. The proof is based on the proof of Theorem 11.8.2 in \cite{dudley2002real}.
\begin{proof}
	%	As we have noticed in the proof of Lemma 8, 
	%	$$\inf_{u\in \mathcal{U}} \int_{I} u d\mu=\inf_{u\in \mathcal{U}_a} \int_{I} u d\mu$$
	%where $\mathcal{U}_a=\{u\in\mathcal{U}: 0\le u\le \norm{a}_\infty \}$.
	
	For any $u\in L^1(I,\mu)$, we define $F_u:\mathcal{G}\rightarrow \mathbb{R}$ by
	$$F_u(\nu)=\int_I u d\nu$$
	Since $I$ is a metric space and $\mu$ is a finite Borel measure, $C_b(I)$ is dense in $L^1(I,\mu)$. Therefore,
	for any $\nu_1,\nu_2\in \mathcal{G}$ and $g\in C_b(I)$, we have 
	$$ |F_u(\nu_1)-F_u(\nu_2)|\le 2\norm{u-g}_1+\int_I g d(\nu_1-\nu_2)\le 2\norm{u-g}_1+\norm{g}_\infty \norm{\nu_1-\nu_2}_0$$
	Therefore, $F_u$ is weak continuous in $\mathcal{G}$. i.e. $F_u\in C(\mathcal{G})$.

	We define $L$ to be the space containing all such function $F_u$. That is, 
	$$L=\{F_u\in C(\mathcal{G}):u\in L^1(I,\mu)\}$$
	Since, for any $c\in\mathbb{R}$, $u_1,u_2\in L^1(I,\mu)$, we have $F_{cu_1+u_2}=cF_{u_1}+F_{u_2}$, $L$ is a linear subspace of $C(\mathcal{G})$.
	
	Next, define $H\subset C(\mathcal{G})$ by 
	$$H=\{F\in C(\mathcal{G}): F(\nu)\ge s(\nu), \forall \nu\in\mathcal{G}\}$$ 
	It is easy to see H is convex. Moreover, H is nonempty since, by \hyperref[B2]{Assumption (B2)}, $F_a\in H$. Consequently, we also have $F_a+1\in int(H)$.	
	
	Now, we define a linear form $r$ on $L\subset C(\mathcal{G})$ by 
	$$r(F_u)=\int_I ud\mu$$
	It is easy to check $r$ is a linear map. By \hyperref[B2]{Assumption (B2)}, $L\cap H\neq \emptyset$ as $F_a\in L\cap H$. Moreover, $r$ is bounded from below in $L\cap H$, since, for any $\gamma\in \Gamma_{\mathcal{G}}$,
	$$r(F_u)=\int_I ud\mu=\int_{\mathcal{G}} \int_I ud\nu d\gamma\ge \int_{\mathcal{G}} s(\nu) d\gamma\ge \left(\inf_{\nu\in\mathcal{G}} s(\nu)\right) \gamma(\mathcal{G})$$
	
	By By \hyperref[B1]{Assumption (B1)}, $s$ is upper-semi continuous on a compact set $\mathcal{G}$, Therefore, $\inf_{\nu\in\mathcal{G}} s(\nu)>-\infty$. In addition, we have 
	$1=\mu(I)=\int_{\mathcal{G}} \nu(I) d\gamma\le \varepsilon \gamma(\mathcal{G})$.
	Therefore, r is bounded from below by $\left(\inf_{\nu\in\mathcal{G}} s(\nu)\right)/\varepsilon$. By Hahn-Banach theorem (Theorem 6.2.11 in \cite{dudley2002real}), r can be extended to a linear functional $\tilde{r}$ on $C(\mathcal{G})$ such that,
	$$\inf_{F\in H} \tilde{r}(F)=\inf_{F\in L\cap H} r(F)$$
	We claim $\tilde{r}$ is a bounded positive functional on $C(\mathcal{G})$. To see the $r$ is positive for any $F\ge 0$ in $C(\mathcal{G})$ and real number $c\ge 0$, we have $\tilde{s}+cF+1\in H$, where $\tilde{s}$ is a continuous approximation of $s$. Note we have
	$$\inf_{F\in H} \tilde{r}(F)=\inf_{F\in H\cap L} r(F)\ge \frac{\inf_{\nu\in\mathcal{G}} s(\nu)}{\varepsilon}$$
	Therefore, by taking c large enough, we get $\tilde{r}(F)\ge 0$.
	To see $\tilde{r}$ is bounded, we note for any $F\in C(\mathcal{G})$, $|\tilde{r}(F)|\le |\tilde{r}(1)|\norm{F}_{\infty}$. 
	
	Hence, by Riesz representation theorem (Theorem 7.4.1 in \cite{dudley2002real}), there exists a positive Borel measure $\rho$ on $\mathcal{G}$ such that $$\tilde{r}(F)=\int_{\mathcal{G}}Fd\rho$$ 
	for any $F\in C(\mathcal{G})$. Furthermore, we show $\rho\in\Gamma_{\mathcal{G}}$. Note $\tilde{r}=r$ on $L$, for any Borel measurable set $A\subset I$, we take $u=\mathbbm{1}_A$, $F_u(\nu)=\nu(A)$ and we have
	$$\int_{\mathcal{G}} \nu(A)d\rho=\tilde{r}(F_u)=r(F_u)=\int_I ud\mu$$
	Therefore, 
	$$\inf_{u\in\mathcal{U}}\int_I u d\mu=\inf_{F\in H\cap L}r(F)=\inf_{F\in H}\tilde{r}(F)=\inf_{F\in H}\int_{\mathcal{G}} Fd\rho= \int_{\mathcal{G}} s d\rho\le \sup_{\gamma\in\Gamma_{\mathcal{G}}} \int_\mathcal{G} s d\gamma$$
	Conversely, for any $u\in L^1(I,\mu)$, we approximate $u$ by step functions such that
	$\int_I u d\nu\ge s(\nu)$ for all $\nu\in \mathcal{G}$. Therefore, for any indicate function $u=1_{A}$, we have
	$$\int_I ud\mu=\int_{\mathcal{G}}\int_I ud\nu d\gamma\ge \int_{\mathcal{G}} s d\gamma$$
	Consequently,  
	$$\sup_{\gamma\in\Gamma_{\mathcal{G}}} \int_\mathcal{G} v d\gamma\le\inf_{u\in\mathcal{U}}\int_I u d\mu$$
\end{proof}

Lastly, to prove the existence of a stable assignment, we prove establish the equivalence relation between duality theorem \hyperref[theorem: infdimduality]{Theorem \ref{theorem: infdimduality}}.

\begin{lemma}
	Any maximizer $\gamma\in \Gamma_{\mathcal{G}}$ in the duality equation is a stable assignment with the imputation given by the minimizer $u\in \mathcal{U}$. 
	
	Conversely, any stable assignment $\gamma$ solves the maximization problem and the imputation $u$ associated with the assignment solves the minimization problem.	
\end{lemma}

\begin{proof}
	Firstly, we suppose $\gamma\in\Gamma_{\mathcal{G}}$ and $u\in \mathcal{U}$ are the solutions of the maximization and minimization problem respectively. 
	By definition, $\int_I u d\nu\ge s(\nu)$ for all groups $\nu\in \mathcal{G}$. In contrast, by Riesz representation theorem,
	$$\int_I u d\mu=\int_{\mathcal{G}} \left(\int_I u d\nu \right)d\gamma = \int_{\mathcal{G}} s d\gamma$$	
	Therefore, $\int_I u d\nu\le s(\nu)$ for $\gamma-$almost all $\nu\in\mathcal{G}$.
	
	Conversely, we take a stable assignment $\gamma\in\Gamma_{\mathcal{G}}$ with its imputation $u\in\mathcal{U}$. By definition,  $\int_I u d\nu\ge s(\nu)$ for all $\nu\in\mathcal{G}$, and $\int_I u d\nu\le v(\nu)$, for $\gamma$-almost all $\nu\in\mathcal{G}$. Therefore, by Riesz representation theorem,
	$$\int_{\mathcal{G}} sd\gamma= \int_{\mathcal{G}} \int_I ud\nu d\gamma= \int_I ud\mu$$
	For any $y\in \mathcal{U}$, we have that $\int_I y d\nu \ge s(\nu)$ for all $\nu\in\mathcal{G}$. As a result, 
	$$\int_I ud\mu=\int_{\mathcal{G}} sd\gamma\le \int_{\mathcal{G}} \int_I yd\nu d\gamma= \int_I yd\mu$$
	That is, $u$ solves the minimization problem. 
	For any $\tau\in\Gamma_{\mathcal{G}}$, 
	$$\int_{\mathcal{G}} sd\tau\le \int_{\mathcal{G}} \int_I u d\nu d\tau=\int_I ud\mu=\int_{\mathcal{G}} sd\gamma$$
	That is, $\gamma$ solves the maximization problem. 
\end{proof}

\end{document}